\documentclass[11pt]{amsart}      
\usepackage{amssymb}
\usepackage{eucal}
\usepackage{amsmath}
\usepackage{amscd}
\usepackage[dvips]{color}
\usepackage{multicol}
\usepackage[all]{xy}           
\usepackage{graphicx}
\usepackage{color}
\usepackage{colordvi}
\usepackage{xspace}
\usepackage{axodraw}

\usepackage[active]{srcltx} 

\begin{document}  

\newcommand{\nc}{\newcommand}
\newcommand{\delete}[1]{}
\nc{\dfootnote}[1]{{}}          
\nc{\ffootnote}[1]{\dfootnote{#1}}
\nc{\mfootnote}[1]{\footnote{#1}} 
\nc{\todo}[1]{\tred{To do:} #1} \delete{
\nc{\mlabel}[1]{\label{#1}}  
\nc{\mcite}[1]{\cite{#1}}  
\nc{\mref}[1]{\ref{#1}}  
}

\nc{\mlabel}[1]{\label{#1}  
{\hfill \hspace{1cm}{\bf{{\ }\hfill(#1)}}}}
\nc{\mcite}[1]{\cite{#1}{{\bf{{\ }(#1)}}}}  
\nc{\mref}[1]{\ref{#1}{{\bf{{\ }(#1)}}}}  

\nc{\mbibitem}[1]{\bibitem{#1}} 
\nc{\mkeep}[1]{\marginpar{{\bf #1}}} 

\newtheorem{theorem}{Theorem}[section]
\newtheorem{prop}[theorem]{Proposition}
\newtheorem{defn}[theorem]{Definition}
\newtheorem{lemma}[theorem]{Lemma}
\newtheorem{coro}[theorem]{Corollary}
\newtheorem{prop-def}[theorem]{Proposition-Definition}
\newtheorem{claim}{Claim}[section]
\newtheorem{remark}[theorem]{Remark}
\newtheorem{propprop}{Proposed Proposition}[section]
\newtheorem{conjecture}{Conjecture}
\newtheorem{exam}[theorem]{Example}
\newtheorem{assumption}{Assumption}
\newtheorem{condition}[theorem]{Assumption}

\renewcommand{\labelenumi}{{\rm(\alph{enumi})}}
\renewcommand{\theenumi}{\alph{enumi}}

\nc{\tred}[1]{\textcolor{red}{#1}}
\nc{\tblue}[1]{\textcolor{blue}{#1}}
\nc{\tgreen}[1]{\textcolor{green}{#1}}
\nc{\tpurple}[1]{\textcolor{purple}{#1}}
\nc{\btred}[1]{\textcolor{red}{\bf #1}}
\nc{\btblue}[1]{\textcolor{blue}{\bf #1}}
\nc{\btgreen}[1]{\textcolor{green}{\bf #1}}
\nc{\btpurple}[1]{\textcolor{purple}{\bf #1}}

\nc{\li}[1]{\textcolor{red}{Li:#1}}
\nc{\cm}[1]{\textcolor{blue}{Chengming: #1}}
\nc{\xiang}[1]{\textcolor{green}{Xiang: #1}}

\nc{\adec}{\check{;}} \nc{\aop}{\alpha}
\nc{\dftimes}{\widetilde{\otimes}} \nc{\dfl}{\succ} \nc{\dfr}{\prec}
\nc{\dfc}{\circ} \nc{\dfb}{\bullet} \nc{\dft}{\star}
\nc{\dfcf}{{\mathbf k}} \nc{\apr}{\ast} \nc{\spr}{\cdot}
\nc{\twopr}{\circ} \nc{\sempr}{\ast}
\nc{\disp}[1]{\displaystyle{#1}}
\nc{\bin}[2]{ (_{\stackrel{\scs{#1}}{\scs{#2}}})}  
\nc{\binc}[2]{ \left (\!\! \begin{array}{c} \scs{#1}\\
    \scs{#2} \end{array}\!\! \right )}  
\nc{\bincc}[2]{  \left ( {\scs{#1} \atop
    \vspace{-.5cm}\scs{#2}} \right )}  
\nc{\sarray}[2]{\begin{array}{c}#1 \vspace{.1cm}\\ \hline
    \vspace{-.35cm} \\ #2 \end{array}}
\nc{\bs}{\bar{S}} \nc{\dcup}{\stackrel{\bullet}{\cup}}
\nc{\dbigcup}{\stackrel{\bullet}{\bigcup}} \nc{\etree}{\big |}
\nc{\la}{\longrightarrow} \nc{\fe}{\'{e}} \nc{\rar}{\rightarrow}
\nc{\dar}{\downarrow} \nc{\dap}[1]{\downarrow
\rlap{$\scriptstyle{#1}$}} \nc{\uap}[1]{\uparrow
\rlap{$\scriptstyle{#1}$}} \nc{\defeq}{\stackrel{\rm def}{=}}
\nc{\dis}[1]{\displaystyle{#1}} \nc{\dotcup}{\,
\displaystyle{\bigcup^\bullet}\ } \nc{\sdotcup}{\tiny{
\displaystyle{\bigcup^\bullet}\ }} \nc{\hcm}{\ \hat{,}\ }
\nc{\hcirc}{\hat{\circ}} \nc{\hts}{\hat{\shpr}}
\nc{\lts}{\stackrel{\leftarrow}{\shpr}}
\nc{\rts}{\stackrel{\rightarrow}{\shpr}} \nc{\lleft}{[}
\nc{\lright}{]} \nc{\uni}[1]{\tilde{#1}} \nc{\wor}[1]{\check{#1}}
\nc{\free}[1]{\bar{#1}} \nc{\den}[1]{\check{#1}} \nc{\lrpa}{\wr}
\nc{\curlyl}{\left \{ \begin{array}{c} {} \\ {} \end{array}
    \right .  \!\!\!\!\!\!\!}
\nc{\curlyr}{ \!\!\!\!\!\!\!
    \left . \begin{array}{c} {} \\ {} \end{array}
    \right \} }
\nc{\leaf}{\ell}       
\nc{\longmid}{\left | \begin{array}{c} {} \\ {} \end{array}
    \right . \!\!\!\!\!\!\!}
\nc{\ot}{\otimes} \nc{\sot}{{\scriptstyle{\ot}}}
\nc{\otm}{\overline{\ot}} \nc{\ora}[1]{\stackrel{#1}{\rar}}
\nc{\ola}[1]{\stackrel{#1}{\la}}
\nc{\pltree}{\calt^\pl} \nc{\epltree}{\calt^{\pl,\NC}}
\nc{\rbpltree}{\calt^r} \nc{\scs}[1]{\scriptstyle{#1}}
\nc{\mrm}[1]{{\rm #1}}
\nc{\dirlim}{\displaystyle{\lim_{\longrightarrow}}\,}
\nc{\invlim}{\displaystyle{\lim_{\longleftarrow}}\,}
\nc{\mvp}{\vspace{0.5cm}} \nc{\svp}{\vspace{2cm}}
\nc{\vp}{\vspace{8cm}} \nc{\proofbegin}{\noindent{\bf Proof: }}
\nc{\proofend}{$\blacksquare$ \vspace{0.5cm}}
\nc{\freerbpl}{{F^{\mathrm RBPL}}}
\nc{\sha}{{\mbox{\cyr X}}}  
\nc{\ncsha}{{\mbox{\cyr X}^{\mathrm NC}}} \nc{\ncshao}{{\mbox{\cyr
X}^{\mathrm NC,\,0}}}
\nc{\shpr}{\diamond}    
\nc{\shprm}{\overline{\diamond}}    
\nc{\shpro}{\diamond^0}    
\nc{\shprr}{\diamond^r}     
\nc{\shpra}{\overline{\diamond}^r} \nc{\shpru}{\check{\diamond}}
\nc{\catpr}{\diamond_l} \nc{\rcatpr}{\diamond_r}
\nc{\lapr}{\diamond_a} \nc{\sqcupm}{\ot} \nc{\lepr}{\diamond_e}
\nc{\vep}{\varepsilon} \nc{\labs}{\mid\!} \nc{\rabs}{\!\mid}
\nc{\hsha}{\widehat{\sha}} \nc{\lsha}{\stackrel{\leftarrow}{\sha}}
\nc{\rsha}{\stackrel{\rightarrow}{\sha}} \nc{\lc}{\lfloor}
\nc{\rc}{\rfloor} \nc{\tpr}{\sqcup} \nc{\nctpr}{\vee}
\nc{\plpr}{\star} \nc{\rbplpr}{\bar{\plpr}} \nc{\sqmon}[1]{\langle
#1\rangle} \nc{\forest}{\calf} \nc{\ass}[1]{\alpha({#1})}
\nc{\altx}{\Lambda_X} \nc{\vecT}{\vec{T}} \nc{\onetree}{\bullet}
\nc{\Ao}{\check{A}} \nc{\seta}{\underline{\Ao}}
\nc{\deltaa}{\overline{\delta}} \nc{\trho}{\tilde{\rho}}

\nc{\mmbox}[1]{\mbox{\ #1\ }} \nc{\ann}{\mrm{ann}}
\nc{\Aut}{\mrm{Aut}} \nc{\can}{\mrm{can}} \nc{\twoalg}{{two-sided
algebra}\xspace} \nc{\colim}{\mrm{colim}} \nc{\Cont}{\mrm{Cont}}
\nc{\rchar}{\mrm{char}} \nc{\cok}{\mrm{coker}} \nc{\dtf}{{R-{\rm
tf}}} \nc{\dtor}{{R-{\rm tor}}}
\renewcommand{\det}{\mrm{det}}
\nc{\depth}{{\mrm d}} \nc{\Div}{{\mrm Div}} \nc{\End}{\mrm{End}}
\nc{\Ext}{\mrm{Ext}} \nc{\Fil}{\mrm{Fil}} \nc{\Frob}{\mrm{Frob}}
\nc{\Gal}{\mrm{Gal}} \nc{\GL}{\mrm{GL}} \nc{\Hom}{\mrm{Hom}}
\nc{\hsr}{\mrm{H}} \nc{\hpol}{\mrm{HP}} \nc{\id}{\mrm{id}}
\nc{\im}{\mrm{im}} \nc{\incl}{\mrm{incl}} \nc{\length}{\mrm{length}}
\nc{\LR}{\mrm{LR}} \nc{\mchar}{\rm char} \nc{\NC}{\mrm{NC}}
\nc{\mpart}{\mrm{part}} \nc{\pl}{\mrm{PL}} \nc{\ql}{{\QQ_\ell}}
\nc{\qp}{{\QQ_p}} \nc{\rank}{\mrm{rank}} \nc{\rba}{\rm{RBA }}
\nc{\rbas}{\rm{RBAs }} \nc{\rbpl}{\mrm{RBPL}} \nc{\rbw}{\rm{RBW }}
\nc{\rbws}{\rm{RBWs }} \nc{\rcot}{\mrm{cot}}
\nc{\rest}{\rm{controlled}\xspace} \nc{\rdef}{\mrm{def}}
\nc{\rdiv}{{\rm div}} \nc{\rtf}{{\rm tf}} \nc{\rtor}{{\rm tor}}
\nc{\res}{\mrm{res}} \nc{\SL}{\mrm{SL}} \nc{\Spec}{\mrm{Spec}}
\nc{\tor}{\mrm{tor}} \nc{\Tr}{\mrm{Tr}} \nc{\mtr}{\mrm{sk}}

\nc{\ab}{\mathbf{Ab}} \nc{\Alg}{\mathbf{Alg}}
\nc{\Algo}{\mathbf{Alg}^0} \nc{\Bax}{\mathbf{Bax}}
\nc{\Baxo}{\mathbf{Bax}^0} \nc{\RB}{\mathbf{RB}}
\nc{\RBo}{\mathbf{RB}^0} \nc{\BRB}{\mathbf{RB}}
\nc{\Dend}{\mathbf{DD}} \nc{\bfk}{{\bf k}} \nc{\bfone}{{\bf 1}}
\nc{\base}[1]{{a_{#1}}} \nc{\detail}{\marginpar{\bf More detail}
    \noindent{\bf Need more detail!}
    \svp}
\nc{\Diff}{\mathbf{Diff}} \nc{\gap}{\marginpar{\bf
Incomplete}\noindent{\bf Incomplete!!}
    \svp}
\nc{\FMod}{\mathbf{FMod}} \nc{\mset}{\mathbf{MSet}}
\nc{\rb}{\mathrm{RB}} \nc{\Int}{\mathbf{Int}}
\nc{\Mon}{\mathbf{Mon}}
\nc{\remarks}{\noindent{\bf Remarks: }}
\nc{\OS}{\mathbf{OS}} 
\nc{\Rep}{\mathbf{Rep}} \nc{\Rings}{\mathbf{Rings}}
\nc{\Sets}{\mathbf{Sets}} \nc{\DT}{\mathbf{DT}}

\nc{\BA}{{\mathbb A}} \nc{\CC}{{\mathbb C}} \nc{\DD}{{\mathbb D}}
\nc{\EE}{{\mathbb E}} \nc{\FF}{{\mathbb F}} \nc{\GG}{{\mathbb G}}
\nc{\HH}{{\mathbb H}} \nc{\LL}{{\mathbb L}} \nc{\NN}{{\mathbb N}}
\nc{\QQ}{{\mathbb Q}} \nc{\RR}{{\mathbb R}} \nc{\TT}{{\mathbb T}}
\nc{\VV}{{\mathbb V}} \nc{\ZZ}{{\mathbb Z}}


\nc{\calao}{{\mathcal A}} \nc{\cala}{{\mathcal A}}
\nc{\calc}{{\mathcal C}} \nc{\cald}{{\mathcal D}}
\nc{\cale}{{\mathcal E}} \nc{\calf}{{\mathcal F}}
\nc{\calfr}{{{\mathcal F}^{\,r}}} \nc{\calfo}{{\mathcal F}^0}
\nc{\calfro}{{\mathcal F}^{\,r,0}} \nc{\oF}{\overline{F}}
\nc{\calg}{{\mathcal G}} \nc{\calh}{{\mathcal H}}
\nc{\cali}{{\mathcal I}} \nc{\calj}{{\mathcal J}}
\nc{\call}{{\mathcal L}} \nc{\calm}{{\mathcal M}}
\nc{\caln}{{\mathcal N}} \nc{\calo}{{\mathcal O}}
\nc{\calp}{{\mathcal P}} \nc{\calr}{{\mathcal R}}
\nc{\calt}{{\mathcal T}} \nc{\caltr}{{\mathcal T}^{\,r}}
\nc{\calu}{{\mathcal U}} \nc{\calv}{{\mathcal V}}
\nc{\calw}{{\mathcal W}} \nc{\calx}{{\mathcal X}}
\nc{\CA}{\mathcal{A}}

\nc{\fraka}{{\mathfrak a}} \nc{\frakB}{{\mathfrak B}}
\nc{\frakb}{{\mathfrak b}} \nc{\frakd}{{\mathfrak d}}
\nc{\oD}{\overline{D}} \nc{\frakF}{{\mathfrak F}}
\nc{\frakg}{{\mathfrak g}} \nc{\frakm}{{\mathfrak m}}
\nc{\frakM}{{\mathfrak M}} \nc{\frakMo}{{\mathfrak M}^0}
\nc{\frakp}{{\mathfrak p}} \nc{\frakS}{{\mathfrak S}}
\nc{\frakSo}{{\mathfrak S}^0} \nc{\fraks}{{\mathfrak s}}
\nc{\os}{\overline{\fraks}} \nc{\frakT}{{\mathfrak T}}
\nc{\oT}{\overline{T}}
\nc{\frakX}{{\mathfrak X}} \nc{\frakXo}{{\mathfrak X}^0}
\nc{\frakx}{{\mathbf x}}
\nc{\frakTx}{\frakT}      
\nc{\frakTa}{\frakT^a}        
\nc{\frakTxo}{\frakTx^0}   
\nc{\caltao}{\calt^{a,0}}   
\nc{\ox}{\overline{\frakx}} \nc{\fraky}{{\mathfrak y}}
\nc{\frakz}{{\mathfrak z}} \nc{\oX}{\overline{X}}

\font\cyr=wncyr10

\nc{\redtext}[1]{\textcolor{red}{#1}}


\title[Special symplectic Lie groups and hypersymplectic
Lie groups]{Special symplectic Lie groups and hypersymplectic Lie groups}

\author{Xiang Ni}
\address{Chern Institute of Mathematics \& LPMC, Nankai
University, Tianjin 300071, P.R.
China}\email{xiangn$_-$math@yahoo.cn}

\author{Chengming Bai}
\address{Chern Institute of Mathematics \& LPMC, Nankai University, Tianjin 300071, P.R. China}
         \email{baicm@nankai.edu.cn}



\begin{abstract}
A special symplectic Lie group is a triple $(G,\omega,\nabla)$ such
that $G$ is a  finite-dimensional real Lie group and $\omega$ is a
left invariant symplectic form on $G$ which is parallel with respect
to a left invariant affine structure $\nabla$. In this paper
starting from a special symplectic Lie group we show  how to
``deform" the standard Lie group structure on the (co)tangent bundle
through the left invariant affine structure $\nabla$ such that the
resulting Lie group admits families of left invariant
hypersymplectic structures and thus becomes a hypersymplectic Lie
group. We consider the affine cotangent extension problem and then
introduce notions of post-affine structure and post-left-symmetric
algebra which is the underlying algebraic structure of a special
symplectic Lie algebra. Furthermore, we give a kind of double
extensions of special symplectic Lie groups in terms of
post-left-symmetric algebras.
\end{abstract}

\subjclass[2000]{17B05, 17D25, 22E20}

\keywords{Special symplectic Lie group, hypersymplectic Lie group,
post-affine structure}

\baselineskip=16pt


\maketitle

\tableofcontents

\setcounter{section}{0} {\ } \vspace{-1cm}

\section{Introduction}
A {\it hypersymplectic structure} on a $4n$-dimensional smooth
manifold $M$ is a triple $\{J,E,g\}$ satisfying the following
conditions:
\begin{enumerate}
\item
$J,E$ are two endomorphisms of the tangent bundle of $M$ such that
$$J^2=-{\rm id},\quad E^2={\rm id},\quad JE=-EJ.$$
\item $g$ is a neutral metric (that is, of signature
$(2n,2n)$) satisfying
$$g(X,Y)=g(JX,JY)=-g(EX,EY),\quad \forall X,Y\in\Gamma(TM).$$
\item
The associated 2-forms
$$\omega_1(X,Y)=g(JX,Y),\quad \omega_2(X,Y)=g(EX,Y),\quad
\omega_3(X,Y)=g(JEX,Y),$$ are closed.
\end{enumerate}
Hypersymplecitc structures were introduced by Hitchin~(\cite{Hit}).
They are the {\it split-quaternion} analogues of hyperk\"{a}hler
structures, where the base manifold carries a {\it hypercomplex}
structure, i.e., a pair $\{J_1,J_2\}$ of commuting complex
structures, and are sometimes called {\it neutral hyperk\"{a}hler}
structures~(\cite{Ka}) or {\it parahyperk\"{a}hler}
structures~(\cite{ABCV}).

Manifolds carrying a hypersymplectic structure have a rich geometry,
such as the (neutral) metric is K\"{a}hler and Ricci flat and its
holonomy group is contained in $Sp(2n, \mathbb{R})$~(\cite{Hit}).
Moreover they are {\it neutral Calabi-Yau}~(\cite{FPPS}). Recently,
hypersymplectic structures are attractive in theoretical physics
since they play a role in string theory~(\cite{OV}). On the other
hand, the relationships between hypersymplectic structures and
integrable systems were considered in~\cite{BM}. The quotient
construction for hypersymplectic manifolds has been studied
in~\cite{Hit} and later in~\cite{DS}. In~\cite{Ka}, Kamada
classified all the hypersymplectic structures on primary Kodaira
surfaces. Examples on 2-step nilmanifolds in higher dimensions were
obtained in~\cite{FPPS}. On the other hand, hypersymplectic
structures on solvable Lie groups have been studied in~\cite{An,
AnD}.

In this paper, we focus on the homogeneous context of
hypersymplectic geometry. More precisely, we study how to construct
hypersymplectic Lie groups. By definition, a {\it hypersymplectic
Lie group} is a real Lie group endowed with a left invariant
hypersymplectic structure. Our starting data is a special symplectic
Lie group. A {\it special symplectic Lie group} is a triple
$(G,\omega,\nabla)$ such that $G$ is a  finite-dimensional real Lie
group and $\omega$ is a left invariant symplectic form on $G$ which
is parallel with respect to a left invariant affine structure
$\nabla$. Recall that an {\it affine structure} on an
$n$-dimensional smooth manifold is specified by a flat and torsion
free connection $\nabla$. If $G$ is a Lie group, an affine structure
is called {\it left invariant} if for each $g\in G$ the
left-multiplication by $g$, $L_g:G\to G$, is an automorphism of the
affine structure. It is known for a long time that left-invariant
affine structures on a connected and simply connected Lie group are
in one-to-one correspondence with the so-called {\it left-symmetric
algebraic structures} on its Lie algebra (see the beginning of
Section~\ref{se:conspe}). Furthermore, special symplectic Lie groups
are also related to the central extensions of Poisson brackets of
hydrodynamic type (\cite{BN}).

In general, we can define a {\it special symplectic manifold} to be
a symplectic manifold with a flat and torsion free connection
$\nabla$ which is parallel with respect to the symplectic form
$\omega$. Note that there is also a notion of special symplectic
manifold which, in addition, involves a complex structure $J$
satisfying the condition that $\nabla J$ is symmetric
(see~\cite{ACD}). Both of these two special symplectic geometries
are the natural underlying geometry of the so-called {\it special
K\"{a}hler geometry} which is quite active recently due to its
relationships with supersymmetric field theories and algebraic
integrable systems (see a survey article~\cite{Co} and the
references therein).

Starting from a special symplectic Lie group, we will show how to
``deform" the standard Lie group structure on its (co)tangent bundle
through the left invariant affine structure $\nabla$ such that the
resulting Lie group admits families of left invariant
hypersymplectic structures and thus become a hypersymplectic Lie
group.

So it is necessary to consider how to construct special symplectic
Lie groups. In order to find some interesting examples of special
symplectic Lie groups, we consider the {\it affine cotangent
extension problem}. To solve such a problem we introduce notions of
{\it post-affine structure} and (its corresponding) {\it
post-left-symmetric algebra}.

We show that the post-left-symmetric algebra is the natural
underlying algebraic structure of a special symplectic Lie algebra
(that is, the Lie algebra of a special symplectic Lie group).
Moreover, there is a commuting diagram which exhibits the role of
post-left-symmetric algebras (see Remark~\ref{re:explain}):
\begin{equation}\begin{matrix} \mbox{left-symmetric algebras}
&\stackrel{}{\leftarrow} & \mbox{post-left-symmetric algebras} \cr
\uparrow\small{} &&\uparrow{} \cr
 \mbox{symplectic Lie algebras} &\stackrel{}{\leftarrow}
& \mbox{special symplectic Lie algebras}\cr
\end{matrix}\label{eq:diagram}
\end{equation}

Furthermore, using post-affine structures and post-left-symmetric
algebras, we provide a kind of constructions of special symplectic
Lie groups (algebras): double extensions. In particular, the theorem
of ``Drinfeld's double" enables us to construct certain interesting
examples of special symplectic Lie groups. Moreover, the double
extensions of special symplectic Lie algebras are the natural
underlying structures of {\it flat} hypersymplectic Lie algebras
(see Remark~\ref{re:flathy}). We also introduce a notion of a very
special para-K\"ahler Lie algebra and investigate its structure.

\medskip

This paper is organized as follows. In section 2, we recall some
basic facts on hypersymplectic Lie groups. We also introduce the
notion of special symplectic Lie group. In section 3, we study the
construction of hypersymplectic Lie groups from tangent and
cotangent bundles of special symplectic Lie groups respectively. In
section 4, we consider the affine cotangent extension problem. We
introduce notions of post-affine structure and post-left-symmetric
algebra and give a kind of double extensions of special symplectic
Lie groups. By a similar study, we investigate the structure of a
very special para-K\"ahler Lie algebra which is an example of
special para-K\"ahler Lie algebra in an appendix.

\smallskip

\noindent {\bf Conventions: }In this paper, the base field is taken
to be $\RR$ of real numbers unless otherwise specified. This is the
field from which we take all the constants and over which we take
all the algebras, vector spaces, linear maps, Lie groups and
manifolds, etc. All algebras, vector spaces, Lie groups and
manifolds are assumed to be finite-dimensional, although many
results still hold in infinite-dimensional case. Let $V$ be a vector
space and $V^*$ be its dual space. We let $\langle,\rangle$ be the
natural pairing between $V$ and $V^*$. In this paper, all the Lie
groups are assumed to be connected and simply connected.


\section{Preliminaries}

In this section, we always let $G$ be a Lie group with Lie algebra
$\frak{g}$.
\subsection{Special symplectic Lie groups}
A {\it left invariant symplectic form} on $G$ (or a {\it symplectic
form} on $\frak{g}$) is defined as a left invariant nondegenerate
2-form which is {\it closed}, that is,
\begin{equation}
d\omega(x,y,z)=\omega(x,[y,z])+\omega(y,[z,x])+\omega(z,[x,y])=0,\;\;
\forall x,y,z\in\frak{g}.
\end{equation}
In this case $G$ is called a {\it symplectic Lie group}~(\cite{Chu,
LiM}) and $\frak{g}$ is called a {\it symplectic Lie algebra}.

A {\it connection} on the Lie algebra $\frak{g}$ is a bilinear map
$\nabla:\frak{g}\times\frak{g}\to\frak{g}$.  After left translating,
$\nabla$ gives rise to a {\it left invariant connection} on the Lie
group $G$, i.e., each left translation $L_g:G\to G, x\to gx$ is an
affine transformation of $G$. In this case, if $x,y\in\frak{g}$ are
two left invariant vector fields on $G$, then $\nabla_xy$ is also
left invariant. The connection $\nabla$ is called {\it torsion free}
if $\nabla_xy-\nabla_yx=[x,y]$ for all $x,y\in\frak{g}$ and is
called {\it flat} if the curvature $R$ of $\nabla$ is identically
zero, where
$R(x,y)=\nabla_x\nabla_y-\nabla_y\nabla_x-\nabla_{[x,y]}$ for any
$x,y\in\frak{g}$.

 A {\it left invariant affine structure} on $G$ is a left
invariant flat and torsion free linear connection on $G$. An
important class of Lie groups having left invariant affine
structures are the symplectic Lie groups~(\cite{Chu}): suppose that
$G$ is a symplectic Lie group with the symplectic form $\omega$,
then there is a left invariant affine structure $\nabla$ on $G$
defined by
\begin{equation}
\omega([x,y],z)=-\omega(y,\nabla_xz),\quad \forall
x,y,z\in\frak{g}.\label{eq:sycoles11}
\end{equation}
\begin{defn}
{\rm A {\it special symplectic Lie group} is a triple
$(G,\nabla,\omega)$ such that
\begin{enumerate}
\item
$G$ is a Lie group;
\item
$\nabla$ is a left invariant affine structure on $G$;
\item
 $\omega$ is a
left invariant symplectic form on $G$ such that it is parallel with
respect to $\nabla$, that is,
\begin{equation}
\omega(\nabla_xy,z)=\omega(\nabla_xz,y),\quad \forall
x,y,z\in\frak{g}.\label{eq:paralsy}
\end{equation}
\end{enumerate}
The Lie algebra of a special symplectic Lie group is called a {\it
special symplectic Lie algebra}.}
\end{defn}
Note that since $\nabla$ is torsion free, Eq.~(\ref{eq:paralsy})
necessarily implies that $\omega$ is closed. So we have the
following conclusion:

\begin{prop}
A special symplectic Lie group $(G,\nabla,\omega)$ is a
finite-dimensional Lie group $G$ equipped with a left invariant
affine structure $\nabla$ and a left invariant nondegenerate 2-form
$\omega$ such that $\omega$ is parallel with respect to
$\nabla$.\label{pp:desplie}
\end{prop}

\begin{exam}{\rm ~(\cite{An})}\quad
{\rm We consider the 2-dimensional special symplectic Lie algebras.
Since the classification of left invariant affine connections (that
is, left-symmetric algebras, see Section 4) on 2-dimensional
connected and simply connected Lie groups (in the sense of
isomorphism over the complex field) was given in \cite{BM1} or
\cite{Bu1}, it is straightforward to get the following 2-dimensional
special symplectic Lie algebras: let $\frak g$ be a 2-dimensional
real Lie algebra with a basis $\{e_1,e_2\}$. Let $\omega$ be the
skew-symmetric bilinear form satisfying
$\omega(e_1,e_2)=-\omega(e_2,e_1)=1$ (that is, $\omega=e_1\wedge
e_2$).
\begin{enumerate}
\item
 $\frak g$ is abelian: there are two left invariant affine
structures $\nabla^1, \nabla^2$ given by
$$\nabla^1_{e_i}e_j=0,\;\;i,j=1,2;\;\;\nabla^2_{e_1}e_1=e_2,\;\;\nabla^2_{e_1}e_2=\nabla^2_{e_2}e_1
=\nabla^2_{e_2}e_2=0,$$ respectively such that both $(\frak
{g},\nabla^1,\omega)$ and $(\frak {g},\nabla^2,\omega)$ are special
symplectic Lie algebras.
\item
$\frak g$ is the non-abelian Lie algebra satisfying $[e_1,e_2]=e_2$:
there are two left invariant affine structures $\nabla^3, \nabla^4$
given by
$$\nabla^3_{e_1}e_1=\nabla^3_{e_1}e_2=0,\;\;\nabla^3_{e_2}e_1=-e_1,\;\;
\nabla^3_{e_2}e_2 =e_2;$$
$$\nabla^4_{e_1}e_1=0,\;\;\nabla^4_{e_1}e_2=\frac{1}{2}e_1,\;\;\nabla^4_{e_2}e_1=-\frac{1}{2}e_1,\;\;
\nabla^4_{e_2}e_2 =e_1+\frac{1}{2}e_2,$$
 respectively such that both $(\frak {g},\nabla^3,\omega)$ and
$(\frak {g},\nabla^4,\omega)$ are special symplectic Lie algebras.
\end{enumerate}
}
\end{exam}

\subsection{Hypersymplectic Lie groups}
A left invariant $(1,1)$ tensor field on $G$ is specified by a
linear endomorphism $N$ of $\frak{g}$. The {\it Nijenhuis torsion}
of $N$ is defined as
\begin{equation}
T(N)(x,y)=[N(x),N(y)]+N^2[x,y]-N([N(x),y]+[x,N(y)]),\quad \forall
x,y\in\mathfrak{g}.
\end{equation}

A {\it left invariant complex structure} on  $G$ (or a {\it complex
structure} on $\frak{g}$) is a linear endomorphism
$J:\frak{g}\to\frak{g}$ such that $J^2=-{\rm id}$ and its Nijenhuis
torsion vanishes, that is, $T(J)=0$.

A {\it left invariant product structure} on  $G$ (or a {\it product
structure} on $\frak{g}$) is a linear endomorphism
$E:\frak{g}\to\frak{g}$ such that $E^2={\rm id}$ (and $E\ne\pm{\rm
id}$) and its Nijenhuis torsion vanishes, that is, $T(E)=0$. Let
$\frak{g}_{+}$ and $\frak{g}_{-}$ be the eigenspaces corresponding
to the eigenvalues $+1$ and $-1$ of $E$, respectively. If ${\rm
dim}\frak{g}_{+}={\rm dim}\frak{g}_{-}$, the product structure $E$
is called a {\it paracomplex structure}. In this case $\frak{g}$ has
even dimension.

A {\it left invariant complex product structure} on the Lie group
$G$ (or a {\it complex product structure} on $\frak{g}$) is a pair
$\{J,E\}$ of a left invariant complex structure $J$ and a left
invariant product structure $E$ satisfying $JE=-EJ$.

Complex product structures on Lie algebras have been studied
in~\cite{AnS} and later in~\cite{NB}. Now we recall their main
properties. The condition $JE=-EJ$ implies that $J$ is an
isomorphism (as vector spaces) between $\frak{g}_{+}$ and
$\frak{g}_{-}$, the eigenspaces corresponding to the eigenvalues
$+1$ and $-1$ of $E$, respectively. Thus, $E$ is in fact a
paracomplex structure on $\frak{g}$. Every complex product structure
on $\frak{g}$ has therefore an associated {\it double Lie algebra}
(or {\it matched pair of Lie algebras})
$(\frak{g},\frak{g}_{+},\frak{g}_{-})$, that is, $\frak{g}_{+}$ and
$\frak{g}_{-}$ are Lie subalgebras of $\frak{g}$ such that
$\frak{g}=\frak{g}_{+}\oplus\frak{g}_{-}$ (as direct sum of vector
spaces). Moreover, we have $\frak{g}_{\mp}=J\frak{g}_{\pm}$, where
$E|_{\frak{g}_{\pm}}=\pm {\rm id}$.

A complex product structure $\{J,E\}$ on $\frak{g}$ determines
uniquely a torsion free connection $\nabla^{CP}$ on $\frak{g}$ such
that $J$ and $E$ are parallel with respect to $\nabla^{CP}$, that
is, $\nabla^{CP}J=\nabla^{CP}E=0$.
\delete{In fact, $\nabla^{CP}$ is
defined as follows: let $(\frak{g},\frak{g}_{+},\frak{g}_{-})$ be
the double Lie algebra associated to the complex product structure
$\{J,E\}$ and let $p_{+}:\frak{g}\to\frak{g}_{+}$ and
$p_{-}:\frak{g}\to\frak{g}_{-}$ be the projections. Then for any
$x,y\in\frak{g}_{+},a,b\in\frak{g}_{-}$, we have
\begin{equation}\hspace{2cm}
\begin{matrix}\nabla^{CP}_xy=-p_{+}(J([x,J(y)])),\;\;
\nabla^{CP}_ab=-p_{-}(J([a,J(b)])),\;\;\cr
\nabla^{CP}_xa=p_{-}([x,a]),\;\;
\nabla^{CP}_ax=p_{+}([a,x]).\cr\end{matrix}\label{eq:connectioncp}
\end{equation}
Moreover, the connection $\nabla^{CP}$ restricts to flat and torsion
free connections on $\frak{g}_{+}$ and $\frak{g}_{-}$, say
$\nabla^{+}$ and $\nabla^{-}$, respectively.}

Now let $\{J,E\}$ be a left invariant complex product structure and
let $g$ be a left-invariant metric on $G$, which is defined by a
nondegenerate symmetric bilinear form $g:\frak{g}\otimes \frak{g}\to
\mathbb{R}$. We will say that $g$ is {\it compatible} with the left
invariant complex product structure if, for all $x,y\in\frak{g}$,
\begin{equation}
g(J(x),J(y))=g(x,y),\quad g(E(x),E(y))=-g(x,y),\quad \forall
x,y\in\frak{g}.\label{eq:compatibleg}
\end{equation}
In terms of $g$ and $\{J,E\}$ we define three left invariant
nondegenerate 2-forms on $G$:
\begin{equation}
\omega_1(x,y)=g(J(x),y),\;\;  \omega_2(x,y)=g(E(x),y),\;\;
\omega_3(x,y)=g(JE(x),y),\;\;\forall x,y\in\frak{g}.
\label{eq:3form}
\end{equation}

\begin{lemma}{\rm ~(\cite{An})}\quad
With the same notations as above, if $g$ is compatible with the
complex product structure $\{J,E\}$, then the 2-forms $\omega_i$
($i=1,2,3$) on $\frak{g}$ given by Eq.~$($\ref{eq:3form}$)$ satisfy
the following properties:
\begin{enumerate}
\item
$\omega_1(x,y)=\omega_1(J(x),J(y))=\omega_1(E(x),E(y))$, for any
$x,y\in\frak{g}$, whence $\omega_1(x,y)=0$ for any
$x\in\frak{g}_{+},y\in\frak{g}_{-}$.
\item
$-\omega_2(x,y)=\omega_2(J(x),J(y))=\omega_2(E(x),E(y))$, for any
$x,y\in\frak{g}$, whence $\omega_2(x,y)=0$ for any
$x,y\in\frak{g}_{+}$ or $x,y\in\frak{g}_{-}$.
\item
$\omega_3(x,y)=-\omega_3(J(x),J(y))=\omega_3(E(x),E(y))$, for any
$x,y\in\frak{g}$, whence $\omega_3(x,y)=0$ for any
$x\in\frak{g}_{+},y\in\frak{g}_{-}$.
\end{enumerate}
\label{le:3formex}
\end{lemma}

\begin{defn}
{\rm Let $G$ be a Lie group with Lie algebra $\frak{g}$. Let
$\{J,E\}$ be a left invariant complex product structure on $G$. Let
$g$ be a left-invariant metric on $G$ which is compatible with
$\{J,E\}$. If the three left invariant nondegenerate 2-forms defined
in Eq.~(\ref{eq:3form}) are closed, that is, they are symplectic
forms, then the triple $\{J,E,g\}$ is called a {\it left invariant
hypersymplectic structure} on $G$ and $G$ is called a {\it
hypersymplectic Lie group}. The Lie algebra of a hypersymplectic Lie
group is called a {\it hypersymplectic Lie algebra}. }
\end{defn}

\begin{prop}
Let $G$ be a Lie group whose Lie algebra is $\frak{g}$. Let
$\{J,E\}$ be a left invariant  complex product structure on $G$.
Suppose that $g$ is a left invariant metric on $G$ which is
compatible with $\{J,E\}$, that is, Eq.~$($\ref{eq:compatibleg}$)$
holds. Define three left invariant nondegenerate 2-forms
$\omega_1,\omega_2,\omega_3:\bigwedge^2\frak{g}\to\mathbb{R}$ by
Eq.~$($\ref{eq:3form}$)$. If $\omega_1$ is closed, then both
$\omega_2$ and $\omega_3$ are closed, too. In this case, $\{J,E,g\}$
is a left invariant hypersymplectic structure on $G$ and $G$ is a
hypersymplectic Lie group.\label{pp:andrada}
\end{prop}
\begin{proof}
It follows from Proposition 5 of~\cite{An}.
\end{proof}

\delete{
\begin{proof}
Let $(\frak{g},\frak{g}_{+},\frak{g}_{-})$ be the double Lie algebra
associated to the complex product structure $\{J,E\}$ and let
$\nabla^{CP}$ be the connection defined by $\{J,E\}$. Suppose that
 $\nabla^{+}$ and $\nabla^{-}$ are the restrictions
of $\nabla^{CP}$ to $\frak{g}_{+}$ and $\frak{g}_{-}$ respectively
and $\omega_{+}$ and $\omega_{-}$ are the restrictions of $\omega_1$
to $\frak{g}_{+}$ and $\frak{g}_{-}$ respectively. The we claim
$\omega_{+}$ and $\omega_{-}$ are parallel with respect to
$\nabla^{+}$ and $\nabla^{-}$ respectively, that is,
\begin{equation}
\nabla^{\pm}\omega_{\pm}=0.\label{eq:parallelpm}
\end{equation}
 In fact, for any
$x,y,u\in\frak{g}_{+}$, we have
\begin{eqnarray*}
d\omega_1(x,y,J(u))&=&
\omega_1(x,[y,J(u)])+\omega_1(y,[J(u),x])+\omega_1(J(u),[x,y])\\
&=&\omega_1(x,-\nabla_{J(u)}^{CP}y)+\omega_1(y,\nabla_{J(u)}^{CP}x)\\
&=&\omega_1(J(x),-\nabla_{J(u)}^{-}J(y))+\omega_1(J(y),\nabla_{J(u)}^{-}J(x)).
\end{eqnarray*}
Since $d\omega_1=0$, we have $\nabla^{-}\omega^{-}=0$. Moreover, for
any $x,u,v\in\frak{g}_{+}$, we have
\begin{eqnarray*}
d\omega_1(x,J(u),J(v))&=&\omega_1(x,[J(u),J(v)])+\omega_1(J(u),[J(v),x])+\omega_1(J(v),
[x,J(u)])\\
&=&\omega_1(J(u),-\nabla_x^{CP}J(v))+\omega_1(J(v),\nabla_x^{CP}J(u))\\
&=&\omega_1(u,-\nabla_x^{+}v)+\omega_1(v,\nabla_x^{+}u).
\end{eqnarray*}
Since $d\omega_1=0$, we have $\nabla^{+}\omega^{+}=0$. Moreover, it
is obvious that
\begin{equation}
\omega_2(x,y)=-\omega_1(J(E(x)),y),\quad
\omega_3(x,y)=\omega_1(E(x),y),\quad \forall
x,y\in\frak{g}.\label{eq:23to1}
\end{equation}
So for any $x,y,u\in \frak{g}_{+}$, we have
\begin{eqnarray*}
d\omega_2(x,y,J(u))&=&\omega_2(x,[y,J(u)])+\omega_2(y,[J(u),x])+\omega_2(J(u),[x,y])\\
&=&\omega_2(x,\nabla^{CP}_yJ(u))+\omega_2(y,-\nabla_x^{CP}J(u))+\omega_2(J(u),[x,y])\\
&=&-\omega_1(J(x),\nabla_yJ(u))+\omega_1(J(y),\nabla_x^{CP}J(u))-\omega_1(u,[x,y])\\
&=&-\omega_1(x,\nabla_y^{+}u)+\omega_1(y,\nabla_x^{+}u)-\omega_1(u,[x,y])\\
&=&\omega_{+}(u,-\nabla^{+}_yx)+\omega_{+}(u,\nabla^{+}_xy)-\omega_{+}(u,[x,y])=0,
\end{eqnarray*}
where the third equality and fifth equality follow from
Eq.~(\ref{eq:23to1}) and Eq.~(\ref{eq:parallelpm}), respectively.
Furthermore, for any $x,u,v\in\frak{g}_{+}$, we have
\begin{eqnarray*}
d\omega_2(x,J(u),J(v))&=&\omega_2(x,[J(u),J(v)])+\omega_2(J(u),[J(v),x])+\omega_2(J(v),[x,J(u)])\\
&=&\omega_2(x,[J(u),J(v)])+\omega_2(J(u),\nabla^{CP}_{J(v)}x)+\omega_2(J(v),-\nabla_{J(u)}^{CP}x)\\
&=&-\omega_1(J(x),[J(u),J(v)])-\omega_1(u,\nabla^{CP}_{J(v)}x)+\omega_1(v,\nabla_{J(u)}^{CP}x)\\
&=&-\omega_{-}(J(x),[J(u),J(v)])-\omega_{-}(J(u),\nabla^{-}_{J(v)}J(x))+\omega_{-}(J(v),\nabla^{-}_{J(u)}J(x))\\
&=&-\omega_{-}(J(x),[J(u),J(v)])-\omega_{-}(J(x),\nabla^{-}_{J(v)}J(u))+\omega_{-}(J(x),\nabla^{-}_{J(u)}J(v))=0,
\end{eqnarray*}
where the third equality and fifth equality follow from
Eq.~(\ref{eq:23to1}) and Eq.~(\ref{eq:parallelpm}), respectively. On
the other hand, it is obvious that if $x,y,z\in\frak{g}_{+}$ or
$x,y,z\in\frak{g}_{-}$, we show that $d\omega_2(x,y,z)=0$. So
$d\omega_2=0$. Next we prove that $d\omega_3=0$. In fact, for any
$x,y,u\in\frak{g}_{+}$, we have
\begin{eqnarray*}
d\omega_3(x,y,J(u))&=&\omega_3(x,[y,J(u)])+\omega_3(y,[J(u),x])+\omega_3(J(u),[x,y])\\
&=&\omega_3(x,-\nabla_{J(u)}^{CP}y)+\omega_3(y,\nabla_{J(u)}^{CP}x)\\
&=&\omega_1(J(x),-\nabla_{J(u)}^{CP}J(y))+\omega_1(J(y),\nabla^{CP}_{J(u)}J(x))\\
&=&\omega_{-}(J(y),-\nabla_{J(u)}^{CP}J(x))+\omega_{-}(J(y),\nabla_{J(u)}^{CP}J(x))=0,
\end{eqnarray*}
where the third and the fourth equality follows from
Eq.~(\ref{eq:23to1}) and Eq.~(\ref{eq:parallelpm}), respectively.
Furthermore, for any $x,u,v\in\frak{g}_{+}$, we have that
\begin{eqnarray*}
d\omega_3(x,J(u),J(v))&=&\omega_3(x,[J(u),J(v)])+\omega_3(J(u),[J(v),x])+\omega_3(J(v),[x,J(u)])\\
&=&\omega_3(J(u),-\nabla_x^{CP}J(v))+\omega_3(J(v),\nabla_x^{CP}J(u))\\
&=&\omega_1(J(u),\nabla_{x}^{CP}J(v))-\omega_1(J(v),\nabla_x^{CP}J(u))\\
&=&\omega_{+}(v,\nabla_{x}^{+}u)-\omega_{+}(v,\nabla_{x}^{+}u)=0,
\end{eqnarray*}
where the third and the fourth equality follows from
Eq.~(\ref{eq:23to1}) and Eq.~(\ref{eq:parallelpm}), respectively. On
the hand, by Eq.~(\ref{eq:23to1}), for any $x,y,z\in\frak{g}_{+}$,
we have
\begin{eqnarray*}
d\omega_3(x,y,z)&=&\omega_3(x,[y,z])+\omega_3(y,[z,x])+\omega_3(z,[x,y])\\
&=&\omega_1(x,[y,z])+\omega_1(y,[z,x])+\omega_1(z,[x,y])=0.
\end{eqnarray*}
Similarly, when $x,y,z\in\frak{g}_{-}$, we show that
$d\omega_3(x,y,z)=0$. So $d\omega_3=0$.
\end{proof}
}

At the end of this section, we recall the semidirect sum of a Lie
algebra and a representation: let $\frak{g}$ be a Lie algebra and
$V$ be a vector space. Let $\rho:\frak{g}\to\frak{gl}(V)$ be a
representation of $\frak{g}$. Then the following bracket operation
makes the direct sum of vector spaces $\frak{g}\oplus V$ be a Lie
algebra (\cite{Sc}):
\begin{equation}
[(x,u),(y,v)]=([x,y],\rho(x)v-\rho(y)u),\quad \forall
x,y\in\frak{g},u,v\in V.
\end{equation}
We denote it by $\frak{g}\ltimes_{\rho}V$, which is called the {\it
semidirect sum} of $\frak{g}$ and $V$ .

\section{Constructions of hypersymplectic Lie groups from the tangent and cotangent bundles of special
symplectic Lie groups}

In this section, we study the construction of hypersymplectic Lie
groups from tangent and cotangent bundles of special symplectic Lie
groups respectively.

\subsection{The tangent bundles of special
symplectic Lie groups}\label{se:tan}

Let $G$ be a Lie group whose Lie algebra is $\frak{g}$. Let $TG$ be
the tangent bundle of $G$ which can be identified with
$G\times\frak{g}$. It is possible to identify $G$ with the zero
section in $TG$ and $\frak{g}$ with the fibre over a neutral element
$(0,e)$ of $TG$. Therefore we identify $T_{0,e}(TG)$ with
$\frak{g}\times\frak{g}$. Now suppose that there exists a left
invariant flat and torsion free connection $\nabla$ on $G$. Since
$\nabla$ is flat, the map
\begin{equation}
\rho_{\nabla}:x\to\nabla_x\in\frak{gl}(\frak{g}),\quad \forall
x\in\frak{g},\label{eq:narepre}
\end{equation}
is a (finite-dimensional) representation of $\frak{g}$. Let
$f_{\nabla}:G\to GL(\frak{g})$ be the corresponding representation
of $G$. Then we can endow $TG$ with the following Lie group
structure:
\begin{equation}
(g_1,x_1)\cdot(g_2,x_2)=(g_1g_2,x_1+f_{\nabla}(g_1)x_2),\quad
\forall g_1,g_2\in G,x_1,x_2\in\frak{g}.\label{eq:liestru}
\end{equation}
It is straightforward to check that the Lie bracket corresponding to
the above Lie group multiplication is
\begin{equation}
[(x_1,y_1),(x_2,y_2)]=([x_1,x_2],\nabla_{x_1}y_2-\nabla_{x_2}y_1),\quad
\forall x_1,x_2,y_1,y_2\in\frak{g},
\end{equation}
that is, the Lie algebra of $TG$ (equipped with the Lie group
structure defined by Eq.~(\ref{eq:liestru})) is
$\frak{g}\ltimes_{\rho_{\nabla}}\frak{g}$. Moreover, it is easy to
show that there is a left invariant flat and torsion free connection
$\tilde{\nabla}$ on $TG$ (equipped with the Lie group structure
defined by Eq.~(\ref{eq:liestru})) defined by (cf.~\cite{BaD, DM})
\begin{equation}
\tilde{\nabla}_{(x,z)}(y,w)=(\nabla_xy,\nabla_xw),\quad \forall
x,y,z,w\in\frak{g}.
\end{equation}

\begin{prop}
With the same conditions and notations as above, the Nijenhuis
torsion of the following left invariant $(1,1)$ tensor field defined
on $TG$ vanishes:
\begin{equation}
N_{\lambda_1,\lambda_2,\lambda_3,\lambda_4}(x,y)=(\lambda_1y+\lambda_2x,\lambda_3x+\lambda_4y),\quad
\forall x,y\in\frak{g},\; \lambda_1,\lambda_2,\lambda_3,\lambda_4\in
\RR.
\end{equation}
\label{pp:nijenva}
\end{prop}

To prove this proposition we need the following  lemma.

\begin{lemma}
Let $G$ be a  Lie group whose Lie algebra is $\frak{g}$. Let $N$ be
a linear transformation of $\frak{g}$ which induces a left invariant
$(1,1)$ tensor field on $G$. If there exists a left invariant
torsion free connection $\nabla$ on $G$ such that $N$ is parallel
with respect to $\nabla$, then the Nijenhuis torsion of $N$
vanishes.\label{le:useful}
\end{lemma}

\begin{proof}
Since $N$ is parallel with respect to $\nabla$, we have that
$N(\nabla_xy)=\nabla_xN(y)$ for any $x$, $y\in\frak{g}$. Moreover,
since $\nabla$ is  torsion free, we show that
\begin{eqnarray*}
[N(x),N(y)]+N^2([x,y])&=&\nabla_{N(x)}N(y)-\nabla_{N(y)}N(x)+N^2(\nabla_xy)-N^2(\nabla_yx)\\
&=&N(\nabla_{N(x)}y)-N(\nabla_yN(x))+N(\nabla_xN(y))-N(\nabla_{N(y)}x)\\
&=&N([N(x),y]+[x,N(y)]),
\end{eqnarray*}
for any $x,y\in\frak{g}$.
\end{proof}

\noindent {\it Proof of Proposition~\ref{pp:nijenva}.} In fact, for
any $x,y,z,w\in\frak{g}$, we show that
\begin{eqnarray*}
N_{\lambda_1,\lambda_2,\lambda_3,\lambda_4}(\tilde{\nabla}_{(x,z)}(y,w))&=&
N_{\lambda_1,\lambda_2,\lambda_3,\lambda_4}(\nabla_xy,\nabla_xw)\\
&=&(\lambda_1\nabla_xw+\lambda_2\nabla_xy,\lambda_3\nabla_xy+\lambda_4\nabla_xw),\\
\tilde{\nabla}_{(x,z)}(N_{\lambda_1,\lambda_2,\lambda_3,\lambda_4}((y,w)))&=&
\tilde{\nabla}_{(x,z)}(\lambda_1w+\lambda_2y,\lambda_3y+\lambda_4w)\\
&=&(\lambda_1\nabla_xw+\lambda_2\nabla_xy,\lambda_3\nabla_xy+\lambda_4\nabla_xw).
\end{eqnarray*}Hence
$N_{\lambda_1,\lambda_2,\lambda_3,\lambda_4}(\tilde{\nabla}_{(x,z)}(y,w))=
\tilde{\nabla}_{(x,z)}(N_{\lambda_1,\lambda_2,\lambda_3,\lambda_4}((y,w)))$.
So the conclusion follows from Lemma~\ref{le:useful}.\hfill $\Box$

\begin{lemma}{\rm (\cite{NB})} Let $V_1$ and $V_2$ be two vector spaces
of the same dimension and $f:V_1\rightarrow V_2$ be an invertible
linear map. Set $V=V_1\oplus V_2$. Define a family of linear maps
$N_{\lambda_1,\lambda_2,\lambda_3,\lambda_4}:V\rightarrow V$ by
\begin{equation}
N_{\lambda_1,\lambda_2,\lambda_3,\lambda_4}(x,a)=(\lambda_1f^{-1}(a)+\lambda_2
x,\lambda_3f(x)+\lambda_4a),
\end{equation}
where $x\in V_1$, $a\in V_2$, $\lambda_i\in\mathbb R$, $i=1,2,3,4$.

$(1)$ $N_{\lambda_1,\lambda_2,\lambda_3,\lambda_4}^2=-{\rm id}$ if
and only if $\lambda_1\ne 0$,
$\lambda_3=\frac{-1-\lambda_2^2}{\lambda_1}$ and
$\lambda_4=-\lambda_2$. In this case,
$N_{\lambda_1,\lambda_2,\lambda_3,\lambda_4}$ is re-parameterized by
$\lambda,\mu$ with $\lambda,\mu\in \mathbb R$  as follows:
\begin{equation}
J_{\lambda,\mu}(x,a)=N_{\lambda,\mu,{{-1-\mu^2}\over{\lambda}},-\mu}((x,a))=(\lambda
f^{-1}(a)+\mu x, {{-1-\mu^2}\over{\lambda}}f(x)-\mu a),\;\;
\lambda\neq 0,\label{eq:jlambdamucom1}
\end{equation} for any $x\in V_1,a\in V_2$.

$(2)$ $N_{\lambda_1,\lambda_2,\lambda_3,\lambda_4}^2={\rm id}$ and
$N_{\lambda_1,\lambda_2,\lambda_3,\lambda_4}\neq\pm{\rm id}$ if and
only if $N_{\lambda_1,\lambda_2,\lambda_3,\lambda_4}$ belongs to one
of the following cases:
\begin{eqnarray*}
&{\rm (F'1)}&\;\; N_{0,1,k,-1}(x,a)=\pm(
x,k f(x)-a);\\
&{\rm (F'2)}&\;\; N_{\hat{k},1,0,-
1}(x,a)=\pm(\hat{k} f^{-1}(a)+x,-a),\quad \hat{k}\neq 0; \\
&{\rm (F'3)}&\;\; N_{k_1,k_2,{{1-k_2^2}\over{k_1}},-k_2}(x,a)=(k_1
f^{-1}(a)+k_2x,{{1-k_2^2}\over{k_1}} f(x)-k_2 a),\;\; k_2^2\neq
1,\;\; k_1\neq 0,
\end{eqnarray*}
for any $x\in V_1, a\in V_2$, where $k,\hat{k},k_1,k_2\in\mathbb R$.
Moreover, the eigenspaces corresponding to the eigenvalue $\pm1$ of
any linear map appearing in ${\rm (F'1)}$, ${\rm (F'2)}$ and ${\rm
(F'3)}$ have the same dimension.

$(3)$ The linear operators in the cases ${\rm (F'1)}$, ${\rm (F'2)}$
and ${\rm (F'3)}$
 anticommute with $J_{\lambda,\mu}$ given by Eq.~$($\ref{eq:jlambdamucom1}$)$ if and only
 if the parameters satisfy the following conditions:

${\rm (F'1)}$\;\; $k=-\frac{2\mu}{\lambda}$;

${\rm (F'2)}$\;\; $\hat{k}=\frac{2\mu\lambda}{1+\mu^2}$, $\mu\neq
0$;

${\rm (F'3)}$\;\;
$k_2=\frac{\mu}{\lambda}k_1\pm\sqrt{1-\frac{k_1^2}{\lambda^2}}$,
$k_1^2\leq
 \lambda^2$, $k_1\neq 0$ and
 $(\mu^2+1)^2k_1^2\neq 4\mu^2\lambda^2$.

 Moreover, they are
 explicitly given as follows:
\begin{eqnarray*} &{\rm (F1)}&
J_{\lambda,\mu}(x,a)=(\lambda f^{-1}(a)+\mu x,
{{-1-\mu^2}\over{\lambda}}f(x)-\mu a),\\
&&E_{\lambda,\mu}(x,a)=\pm(x,-\frac{2\mu}{\lambda}f(x)-a),\quad \lambda\neq 0;\\
&{\rm (F2)}& J_{\lambda,\mu}(x,a)=(\lambda f^{-1}(a)+\mu x,
{{-1-\mu^2}\over{\lambda}} f(x)-\mu a),\\
&& \hat{E}_{\lambda,\mu}(x,a)=\pm(\frac{2\mu\lambda}{1+\mu^2}
f^{-1}(a)+x,-a),\quad \lambda\neq 0,\quad \mu\neq 0;
\end{eqnarray*}
\begin{eqnarray*}
&{\rm (F3)}& J_{\lambda,\mu}(x,a)=(\lambda f^{-1}(a)+\mu x,
{{-1-\mu^2}\over{\lambda}} f(x)-\mu a),\\
& &E_{k,\lambda,\mu}^{\pm}(x,a)=(k
f^{-1}(a)+(\frac{k\mu}{\lambda}\pm\sqrt{\frac{1-k^2}{\lambda^2}})x,
(\frac{(1-\mu^2)k}{\lambda^2}\mp\frac{2\mu}{\lambda}\sqrt{1-\frac{k^2}{\lambda^2}})f(x)\\
&&\hspace{3cm}+
(\frac{k\mu}{\lambda}\mp\sqrt{1-\frac{k^2}{\lambda^2}})a),k^2\leq
 \lambda^2, \lambda\neq 0, k\neq 0,
(\mu^2+1)^2k^2\neq 4\mu^2\lambda^2,
\end{eqnarray*}
for any $x\in V_1,a\in V_2$ and $k,\lambda,\mu\in\mathbb
R$.\label{le:nbli}
\end{lemma}

By Proposition~\ref{pp:nijenva} and Lemma~\ref{le:nbli} we have the
following conclusion:

\begin{coro}
With the conditions and notations in Proposition~\ref{pp:nijenva},
there exist three families of left invariant complex product
structures on $TG$ which are given by {\rm (F1)}, {\rm (F2)} and
{\rm (F3)} defined in Lemma~\ref{le:nbli} respectively, where
$V=\frak{g}\ltimes_{\rho_{\nabla}}\frak{g}, V_1=V_2=\frak{g},f={\rm
id}, x,a\in\frak{g}$.\label{co:compro1}
\end{coro}

\begin{remark}
{\rm The left invariant complex structure $J_{1,0}$ (on
$\frak{g}\ltimes_{\rho_{\nabla}}\frak{g}$) has already been known
for a long time (cf.~\cite{AnS, BaD, DM}). In this case the product
structure, which was considered in~\cite{AnS, BaD}, is given by
$E(x,y)=(-x,y)$, for any $x,y\in\frak{g}$. }
\end{remark}

Now we suppose that $\omega:\bigwedge^2\frak{g}\to\mathbb{R}$ is a
left invariant symplectic form on $G$ which is parallel with respect
to $\nabla$, that is, $(G,\nabla,\omega)$ is a special symplectic
Lie group. Define a bilinear form
$g:(\frak{g}\ltimes_{\rho_{\nabla}}\frak{g})\otimes(\frak{g}\ltimes_{\rho_{\nabla}}\frak{g})\to\mathbb{R}$
by
\begin{equation}
g((x,z),(y,w))=\omega(x,w)+\omega(y,z),\quad \forall
x,y,z,w\in\frak{g}.\label{eq:metric1}
\end{equation}

It is easy to show that $g$ is symmetric and since $\omega$ is
nondegenerate, $g$ is also nondegenerate. So after left translating,
$g$ becomes a left invariant (neutral) metric on $TG$. Moreover, for
any $x,y,z,w\in\frak{g}$, we have

{\small \begin{eqnarray*}
&&g(N_{\lambda_1,\lambda_2,\lambda_3,\lambda_4}(x,z),(y,w))+
g((x,z),N_{\lambda_1,\lambda_2,\lambda_3,\lambda_4}(y,w))\\
&=&g((\lambda_1z+\lambda_2x,\lambda_3x+\lambda_4z),(y,w))+g((x,z),(\lambda_1w+\lambda_2y,
\lambda_3y+\lambda_4w))\\
&=&\lambda_1\omega(z,w)+\lambda_2\omega(x,w)+\lambda_3\omega(y,x)+\lambda_4\omega(y,z)+
\lambda_3\omega(x,y)+\lambda_4\omega(x,w)+\lambda_1\omega(w,z)+\lambda_2\omega(y,z)\\
&=&(\lambda_2+\lambda_4)\omega(x,w)+(\lambda_2+\lambda_4)\omega(y,z)=(\lambda_2+\lambda_4)g((x,z),(y,w)).
\end{eqnarray*}}
It is easy to show that the left invariant complex product
structures constructed in Corollary~\ref{co:compro1} are all
compatible with respect to $g$ defined by Eq.~(\ref{eq:metric1})
since in these cases $\lambda_2+\lambda_4=0$. Thus, we define three
left invariant 2-forms
$\omega_1,\omega_2,\omega_3:\bigwedge^2(\frak{g}\ltimes_{\rho_{\nabla}}\frak{g})\to\mathbb{R}$
(on $TG$) through Eq.~(\ref{eq:3form}). In particular, for any
$x,y,z,w\in\frak{g}$, we have
\begin{eqnarray*}
\omega_1((x,z),(y,w))&=&g(J_{\lambda,\mu}(x,z),(y,w))=g((\lambda
z+\mu
x,\frac{-1-\mu^2}{\lambda}x-\mu z),(y,w))\\
&=&\lambda\omega(z,w)+\mu\omega(x,w)-\frac{1+\mu^2}{\lambda}\omega(y,x)-\mu\omega(y,z).
\end{eqnarray*}
Furthermore, for any $x,y,z,u,v,w\in\frak{g}$, we have that

{\small \begin{eqnarray*}
&&\omega_1((x,u),[(y,v),(z,w)])+\omega_1((y,v),[(z,w),(x,u)])+\omega_1((z,w),[(x,u),(y,v)])\\
&=&\omega_1((x,u),([y,z],\nabla_yw-\nabla_zv))+\omega_1((y,v),([z,x],\nabla_zu-\nabla_xw))+
\omega_1((z,w),([x,y],\nabla_xv-\nabla_yu))\\
&=&\lambda\omega(u,\nabla_yw)-\lambda\omega(u,\nabla_zv)+\mu\omega(x,\nabla_yw)-
\mu\omega(x,\nabla_zv)-\frac{1+\mu^2}{\lambda}\omega([y,z],x)-\mu\omega([y,z],u)\\
& &+
\lambda\omega(v,\nabla_zu)-\lambda\omega(v,\nabla_xw)+\mu\omega(y,\nabla_zu)-\mu\omega(y,\nabla_xw)
-\frac{1+\mu^2}{\lambda}\omega([z,x],y)-\mu\omega([z,x],v)\\
&
&+\lambda\omega(w,\nabla_xv)-\lambda\omega(w,\nabla_yu)+\mu\omega(z,\nabla_xv)-\mu\omega(z,\nabla_yu)
-\frac{1+\mu^2}{\lambda}\omega([x,y],z)-\mu\omega([x,y],w)\\
&=&\lambda\omega(u,\nabla_yw)-\lambda\omega(w,\nabla_yu)-\lambda\omega(u,\nabla_zv)+
\lambda\omega(v,\nabla_zu)+\mu\omega(x,\nabla_yw)-\mu\omega(y,\nabla_xw)-\\
& &\mu\omega([x,y],w)
-\mu\omega(x,\nabla_zv)-\mu\omega([z,x],v)+\mu\omega(z,\nabla_xv)-\frac{1+\mu^2}{\lambda}\omega([y,z],x)
-\frac{1+\mu^2}{\lambda}\omega([z,x],y)\\
& &-\frac{1+\mu^2}{\lambda}\omega([x,y],z)-\mu\omega([y,z],u)
+\mu\omega(y,\nabla_zu)-\mu\omega(z,\nabla_yu)-\lambda\omega(v,\nabla_xw)+\lambda\omega(w,\nabla_xv)=0.
\end{eqnarray*}}
Therefore, $\omega_1$ is closed. So by Proposition~\ref{pp:andrada}
we obtain the main result of this subsection:

\begin{theorem}
Let $(G,\nabla,\omega)$ be a special symplectic Lie group. Endow the
tangent bundle $TG$ of $G$ with the Lie group structure defined by
Eq.~$($\ref{eq:liestru}$)$. Then there exist three families of left
invariant hypersymplectic structures $\{{\rm (F1)},g\}$, $\{{\rm
(F2)},g\}$ and  $\{{\rm (F3)},g\}$ on $TG$, where $g$ is defined by
Eq.~$($\ref{eq:metric1}$)$ and the complex product structures are
defined in Lemma~\ref{le:nbli}, where
$V=\frak{g}\ltimes_{\rho_{\nabla}}\frak{g}, V_1=V_2=\frak{g},f={\rm
id}, x,a\in\frak{g}$. In these cases, $TG$ becomes a hypersymplectic
Lie group.
\end{theorem}

\subsection{The cotangent bundles of
special symplectic Lie groups}\label{se:cotan} Let $G$ be a Lie
group whose Lie algebra is $\frak{g}$. Let $T^*G$ be the cotangent
bundle of $G$ which is identified with $G\times\frak{g}^*$. It is
possible to identify $G$ with the zero section in $T^*G$ and
$\frak{g}^*$ with the fibre over a neutral element $(0,e)$ of
$T^*G$. Therefore we identify $T_{0,e}(T^*G)$ with
$\frak{g}\times\frak{g}^*$. Now suppose that there exists a left
invariant flat and torsion free connection $\nabla$ on $G$. Endow
$T^*G$ with the following Lie group structure:
\begin{equation}
(g_1,a^*)\cdot(g_2,b^*)=(g_1g_2,a^*+b^*\circ
f_{\nabla}(g_1^{-1})),\quad \forall g_1,g_2\in G,
a^*,b^*\in\frak{g}^*,\label{eq:coliestru}
\end{equation}
where $f_{\nabla}:G\to GL(\frak{g})$ is the Lie group homomorphism
corresponding to the Lie algebra homomorphism $\rho_{\nabla}$
defined by Eq.~(\ref{eq:narepre}) and $\langle b^*\circ
f_{\nabla}(g_1^{-1}),x\rangle=\langle
b^*,f_{\nabla}(g_1^{-1})x\rangle$, for any $x\in\frak{g}$. Then
according to Boyom~(\cite{Bo1, Bo2}), the Lie bracket corresponding
to the above Lie group multiplication is given by
\begin{equation}
[(x,a^*),(y,b^*)]=([x,y],a^*\circ\nabla_y-b^*\circ\nabla_x),\quad
\forall x,y\in\frak{g}, a^*,b^*\in\frak{g}^*,
\end{equation}
that is, the Lie algebra of $T^*G$ equipped with the Lie group
structure defined by Eq.~(\ref{eq:coliestru}) is
$\frak{g}\ltimes_{\rho_{\nabla}^*}\frak{g}^*$. Moreover, according
to~\cite{Bai1, Bai2}, the following equation defines a left
invariant flat and torsion free connection $\hat{\nabla}$ on $T^*G$
equipped with the Lie group structure defined by
Eq.~(\ref{eq:coliestru}):
\begin{equation}
\hat{\nabla}_{(x,a^*)}(y,b^*)=(\nabla_xy,-b^*\circ\nabla_x),\quad
\forall x\in\frak{g},a^*\in\frak{g}^*.\label{eq:speconn}
\end{equation}
Now suppose that in addition, there exists a left invariant
symplectic form on $G$ with is parallel with respect to $\nabla$,
that is, $(G,\nabla,\omega)$ is a special symplectic Lie group. Then
$\omega$ induces a linear isomorphism
$\varphi:\frak{g}\to\frak{g}^*$ through
\begin{equation}
\omega(x,y)=\langle\varphi(x),y\rangle,\quad \forall
x,y\in\frak{g}.\label{eq:delimap}
\end{equation}
\begin{prop}
With the same conditions and notations as above, the Nijenhuis
torsion of the following left invariant $(1,1)$ tensor field on
$T^*G$ which is equipped with the Lie group structure defined by
Eq.~$($\ref{eq:coliestru}$)$ vanishes:
\begin{equation}
N_{\lambda_1,\lambda_2,\lambda_3,\lambda_4}(x,a^*)=(\lambda_1\varphi^{-1}(a^*)+\lambda_2
x,\lambda_3\varphi(x)+\lambda_4a^*), \forall
x\in\frak{g},a^*\in\frak{g}^*,\forall
\lambda_1,\lambda_2,\lambda_3,\lambda_4\in \RR.
\end{equation}
\label{pp:conijenva}
\end{prop}
\begin{proof}
Since $\omega$ is parallel with respect to $\nabla$, for any
$x,y,z\in\frak{g}$, we have
$$\omega(\nabla_xy,z)+\omega(y,\nabla_xz)=0\Leftrightarrow\langle\varphi(\nabla_xy),z\rangle+\langle
\varphi(y),\nabla_xz\rangle=0\Leftrightarrow\langle\varphi(\nabla_xy),z\rangle+
\langle\varphi(y)\circ\nabla_x,z\rangle=0.$$ Hence
\begin{equation}
\varphi(\nabla_xy)=\varphi(y)\circ\nabla_x,\quad\forall
x,y\in\frak{g}.\label{eq:omegatovar}
\end{equation}
On the other hand, for any $x,y\in\frak{g},a^*,b^*\in\frak{g}^*$, we
have
\begin{eqnarray*}
N_{\lambda_1,\lambda_2,\lambda_3,\lambda_4}(\hat{\nabla}_{(x,a^*)}(y,b^*))&=&
N_{\lambda_1,\lambda_2,\lambda_3,\lambda_4}(\nabla_xy,-b^*\circ
\nabla_x)\\
&=&(-\lambda_1\varphi^{-1}(b^*\circ\nabla_x)+\lambda_2\nabla_xy,\lambda_3\varphi(\nabla_xy)-
\lambda_4b^*\circ \nabla_x),\\
\hat{\nabla}_{(x,a^*)}N_{\lambda_1,\lambda_2,\lambda_3,\lambda_4}(y,b^*)&=&
\hat{\nabla}_{(x,a^*)}(\lambda_1\varphi^{-1}(b^*)+\lambda_2y,\lambda_3\varphi(y)+\lambda_4b^*)\\
&=&(\lambda_1\nabla_x\varphi^{-1}(b^*)+\lambda_2\nabla_xy,-\lambda_3\varphi(y)\circ\nabla_x-\lambda_4
b^*\circ\nabla_x).
\end{eqnarray*}
Therefore, using Eq.~(\ref{eq:omegatovar}) we show that
$$N_{\lambda_1,\lambda_2,\lambda_3,\lambda_4}(\hat{\nabla}_{(x,a^*)}(y,b^*))=
\hat{\nabla}_{(x,a^*)}N_{\lambda_1,\lambda_2,\lambda_3,\lambda_4}(y,b^*).$$
Since $\hat{\nabla}$ is torsion free,  the conclusion follows from
Lemma~\ref{le:useful}.
\end{proof}

Combining Proposition~\ref{pp:conijenva} and Lemma~\ref{le:nbli} we
have the following conclusion:

\begin{coro}
With the conditions and notations in Proposition~\ref{pp:conijenva},
there exist three families of left invariant complex product
structures on $T^*G$ which are given by {\rm (F1)}, {\rm (F2)} and
{\rm (F3)} defined in Lemma~\ref{le:nbli}, where
$V=\frak{g}\ltimes_{\rho_{\nabla}^*}\frak{g}^*,
V_1=\frak{g},V_2=\frak{g}^*,f=\varphi,
x\in\frak{g},a\in\frak{g}^*$.\label{co:compro2}
\end{coro}

 On the other hand, there exists a natural symmetric
and nondegenerate bilinear form
$g:(\frak{g}\ltimes_{\rho_{\nabla}^*}\frak{g}^*)\otimes
(\frak{g}\ltimes_{\rho_{\nabla}^*}\frak{g}^*)\rightarrow \RR$ on
$\frak{g}\ltimes_{\rho_{\nabla}^*}\frak{g}^*$ which induces a left
invariant (neutral) metric on $T^*G$:
\begin{equation}
g((x,a^*),(y,b^*))=\langle x,b^*\rangle+\langle a^*,y\rangle,\quad
\forall x,y\in\frak{g},a^*,b^*\in\frak{g}^*.\label{eq:metric2}
\end{equation}
Moreover, since $\omega$ is skew-symmetric, for any
$x,y\in\frak{g},a^*,b^*\in\frak{g}^*$, we have that
\begin{eqnarray*}
& &g(N_{\lambda_1,\lambda_2,\lambda_3,\lambda_4}(x,a^*),(y,b^*))+
g((x,a^*),N_{\lambda_1,\lambda_2,\lambda_3,\lambda_4}(y,b^*))\\
&=&g((\lambda_1\varphi^{-1}(a^*)+\lambda_2
x,\lambda_3\varphi(x)+\lambda_4a^*),(y,b^*))+g((x,a^*),(\lambda_1\varphi^{-1}(b^*)
+\lambda_2y,\lambda_3\varphi(y)+\lambda_4b^*))\\
&=&\lambda_1\langle\varphi^{-1}(a^*),b^*\rangle+\lambda_2\langle
x,b^*\rangle+\lambda_3\langle\varphi(x),y\rangle+\lambda_4\langle
a^*,y\rangle+\lambda_3\langle x,\varphi(y)\rangle+\lambda_4\langle
x,b^*\rangle+\\ & &\lambda_1\langle
a^*,\varphi^{-1}(b^*)\rangle+\lambda_2\langle a^*,y\rangle\\
&=&(\lambda_2+\lambda_4)\langle
x,b^*\rangle+(\lambda_2+\lambda_4)\langle
a^*,y\rangle=(\lambda_2+\lambda_4)g((x,a^*),(y,b^*)).
\end{eqnarray*}
It is easy to show that the complex product structures constructed
in Corollary~\ref{co:compro2} are all compatible with respect to $g$
defined by Eq.~(\ref{eq:metric2}) since in these cases
$\lambda_2+\lambda_4=0$. Thus, we define three left invariant
2-forms
$\omega_1,\omega_2,\omega_3:\bigwedge^2(\frak{g}\ltimes_{\rho_{\nabla}^*}\frak{g}^*)\to\mathbb{R}$
(on $T^*G$) through Eq.~(\ref{eq:3form}). In particular, for any
$x,y\in\frak{g},a^*,b^*\in\frak{g}^*$, we have that
\begin{eqnarray*}
\omega_1((x,a^*),(y,b^*))&=&g(J_{\lambda,\mu}(x,a^*),(y,b^*))=g((\lambda\varphi^{-1}(a^*)+\mu
x,\frac{-1-\mu^2}{\lambda}\varphi(x)-\mu a^*),(y,b^*))\\
&=&\lambda\langle\varphi^{-1}(a^*),b^*\rangle+\mu\langle
x,b^*\rangle-\frac{1+\mu^2}{\lambda}\langle\varphi(x),y\rangle-\mu\langle
a^*,y\rangle.
\end{eqnarray*}
Furthermore, for any $x,y,z\in\frak{g},a^*,b^*,c^*\in\frak{g}^*$, we
have \allowdisplaybreaks{\begin{eqnarray*} &
&\omega_1((x,a^*),[(y,b^*),(z,c^*)])+\omega_1((y,b^*),[(z,c^*),(x,a^*)])+\omega_1((z,c^*),[(x,a^*),(y,b^*)])\\
&=&\omega_1((x,a^*),([y,z],-c^*\circ\nabla_y+b^*\circ\nabla_z))+
\omega_1((y,b^*),([z,x],-a^*\circ\nabla_z+c^*\circ\nabla_x))\\
& &+\omega_1((z,c^*),
([x,y],-b^*\circ\nabla_x+a^*\circ\nabla_y))\\
&=&-\lambda\langle\nabla_y\varphi^{-1}(a^*),c^*\rangle+\lambda\langle\nabla_z\varphi^{-1}(a^*),b^*\rangle-
\mu\langle\nabla_yx,c^*\rangle+\mu\langle\nabla_zx,b^*\rangle-
\frac{1+\mu^2}{\lambda}\langle\varphi(x),[y,z]\rangle\\
& &-\mu\langle
a^*,[y,z]\rangle-\lambda\langle\nabla_z\varphi^{-1}(b^*),a^*\rangle+\lambda\langle\nabla_x\varphi^{-1}(b^*),c^*\rangle-
\mu\langle\nabla_zy,a^*\rangle+\mu\langle\nabla_xy,c^*\rangle-\\
& &\frac{1+\mu^2}{\lambda}\langle\varphi(y),[z,x]\rangle-\mu\langle
b^*,[z,x]\rangle-\lambda\langle\nabla_x\varphi^{-1}(c^*),b^*\rangle+\lambda\langle\nabla_y\varphi^{-1}(c^*),a^*\rangle-
\mu\langle\nabla_xz,b^*\rangle\\ & &+\mu\langle\nabla_yz,a^*\rangle-
\frac{1+\mu^2}{\lambda}\langle\varphi(z),[x,y]\rangle-\mu\langle
c^*,[x,y]\rangle\\
&=&\lambda\omega(\nabla_y\varphi^{-1}(a^*),\varphi^{-1}(c^*))-
\lambda\omega(\nabla_y\varphi^{-1}(c^*),\varphi^{-1}(a^*))-\lambda\omega(\nabla_z\varphi^{-1}(a^*),\varphi^{-1}(b^*))
+\\
 & &\lambda\omega(\nabla_z\varphi^{-1}(b^*),\varphi^{-1}(a^*))-\mu\langle\nabla_yx,c^*\rangle+
\mu\langle\nabla_xy,c^*\rangle-\mu\langle
c^*,[x,y]\rangle+\mu\langle\nabla_zx,b^*\rangle-\\ & &\mu\langle
b^*,[z,x]\rangle-\mu\langle\nabla_xz,b^*\rangle-\frac{1+\mu^2}{\lambda}\omega(x,[y,z])-
\frac{1+\mu^2}{\lambda}\omega(y,[z,x])-\frac{1+\mu^2}{\lambda}\omega(z,[x,y])\\
& &-\mu\langle a^*,[y,z]\rangle-
\mu\langle\nabla_zy,a^*\rangle+\mu\langle
\nabla_yz,a^*\rangle-\lambda\omega(\nabla_x\varphi^{-1}(b^*),\varphi^{-1}(c^*))\\
&&+ \lambda\omega(\nabla_x\varphi^{-1}(c^*),\varphi^{-1}(b^*))=0.
\end{eqnarray*}}
Therefore, $\omega_1$ is closed. So by Proposition~\ref{pp:andrada},
 we obtain the main result of this
subsection:

\begin{theorem}
Let $(G,\nabla,\omega)$ be a special symplectic Lie group. Endow the
cotangent bundle $T^*G$ of $G$ with the Lie group structure defined
by Eq.~$($\ref{eq:coliestru}$)$. Then there exist three families of
left invariant hypersymplectic structures $\{{\rm (F1)},g\}$,
$\{{\rm (F2)},g\}$ and  $\{{\rm (F3)},g\}$ on $TG$, where $g$ is
defined by Eq.~$($\ref{eq:metric2}$)$ and the complex product
structures are defined in Lemma~\ref{le:nbli}, where
$V=\frak{g}\ltimes_{\rho_{\nabla}^*}\frak{g}^*,
V_1=\frak{g},V_2=\frak{g}^*,f=\varphi, x\in\frak{g},a\in\frak{g}^*$.
In these cases $T^*G$ becomes a hypersymplectic Lie group.
\end{theorem}

Note that these hypersymplectic structures constructed in
Section~\ref{se:tan} and Section~\ref{se:cotan} may be isomorphic in
the sense of~\cite{An}.

\section{Constructions of special symplectic Lie groups}\label{se:conspe}
From the study in the previous section, we show that, in order to
get some interesting hypersymplectic Lie groups along our approach,
the first step might be to construct some examples of special
symplectic Lie groups, which is the content of this section.

\smallskip

\noindent {\bf Conventions: }   Let $V$ be a vector space and
$\diamond:V\otimes V\to V$ be a bilinear operation. We use
$L_{\diamond},R_{\diamond}, L_{\diamond}^*$ and $R_{\diamond}^*$ to
denote the following operations:
\begin{equation}
L_{\diamond}(x)y=x\diamond y,\quad R_{\diamond}(x)y=y\diamond
x,\quad \forall x,y\in V.
\end{equation}
\begin{equation}
\langle L_{\diamond}^*(x)a^*,y\rangle=-\langle a^*,x\diamond
y\rangle,\quad \langle R_{\diamond}^*(x)a^*,y\rangle=-\langle
a^*,y\diamond x\rangle,\quad \forall x,y\in V,a^*\in
V^*.\label{eq:starlrno1}
\end{equation}
If $\frak{g}$ is a Lie algebra, then we use ${\rm ad}$ and ${\rm
ad}^*$ to denote the {\it adjoint action} and {\it coadjoint action
respectively}, that is,
\begin{equation}
{\rm ad}(x)y=[x,y],\quad \langle{\rm ad}^*(x)a^*,y\rangle=-\langle
a^*,[x,y]\rangle,\quad \forall x,y\in \frak{g},a^*\in\frak{g}^*.
\end{equation}

We first recall the notion of a left-symmetric algebra.

\begin{defn} {\rm A {\it left-symmetric algebra $($LSA$)$} (or a {\it
pre-Lie algebra}) is a vector space $A$ equipped with a bilinear
operation $(x,y)\to x\cdot y$ satisfying
\begin{equation} (x\cdot y)\cdot z-x\cdot(y\cdot z)=(y\cdot
x)\cdot z-y\cdot(x\cdot z),\;\;\forall x,y,z\in
A.\label{eq:axiomlsa}
\end{equation}
}
\end{defn}
\begin{prop-def} {\rm (\cite{Bu})}\quad Suppose that $(A, \cdot)$ is an LSA.
\begin{enumerate}
\item The commutator
\begin{equation} [x,y] =
x\cdot y - y\cdot x,\;\;\forall x,y\in A,
\end{equation} defines a Lie algebra $\frak g(A)$, which is called the sub-adjacent
 Lie algebra of $A$ and $A$ is called a compatible LSA
 structure on the Lie algebra $\frak g (A)$.
\item
$L_{\cdot}:A\to\frak{gl}(A)$ gives a representation of the Lie
algebra $\frak g(A)$, that is,
\begin{equation}
L_\cdot({[x, y]}) = L_\cdot(x)L_\cdot(y)- L_\cdot(y)L_\cdot(x),\
\forall x, y\in A.
\end{equation}
\end{enumerate}
\label{pp:burde}
\end{prop-def}
Now let $G$ be a connected and simply connected Lie group whose Lie
algebra is $\frak{g}$. Suppose that there exists a compatible LSA
structure on $\frak{g}$. We can define a connection on $\frak{g}$ by
\begin{equation}
\nabla_xy=x\cdot y,\quad \forall x,y\in\frak{g}.\label{eq:delsa}
\end{equation}
Then Eq.~(\ref{eq:axiomlsa}) translates into the flatness of
$\nabla$. So after left translating, it induces a left invariant
affine structure on $G$. Conversely, if $\nabla$ is a left invariant
affine structure, then it is easy to show that Eq.~(\ref{eq:delsa})
defines a compatible LSA structure on $\frak{g}$. Therefore we have
the following one-to-one correspondence~(\cite{Bu}):
\begin{equation}
\xymatrix{ \text{\{left invariant affine structures on $G$\}}
\ar@{<->}[r]& \text{\{compatible LSA structures on $\frak{g}$\}}
}.\label{eq:onetoone}
\end{equation}

\subsection{The motivation}
Our main motivation to construct some examples of special symplectic
Lie groups comes from the following construction of symplectic Lie
groups due to Boyom~(\cite{Bo1, Bo2}): in fact, Boyom observed that
if $G$ is a Lie group whose Lie algebra is $\frak{g}$ and $\nabla$
is a left invariant flat and torsion free connection on $G$, then
the following bilinear form on $\frak{g}\oplus\frak{g}^*$ induces a
left invariant symplectic form on $T^*G$ which is equipped with the
``deformed" Lie group structure defined by Eq.~(\ref{eq:coliestru}):
\begin{equation}
\omega_p((x,a^*),(y,b^*))=-\langle x,b^*\rangle+\langle
a^*,y\rangle,\;\forall
x,y\in\frak{g},a^*,b^*\in\frak{g}^*.\label{eq:syform}
\end{equation}

\begin{prop}
Let $G$ be a  Lie group whose Lie algebra is $\frak{g}$. Suppose
that there exists a left invariant flat and torsion free connection
$\nabla$ on $G$. Then $(T^*G,\hat{\nabla},\omega_p)$ is a special
symplectic Lie group, where $T^*G$ is endowed with the Lie group
structure defined by Eq.~$($\ref{eq:coliestru}$)$, $\hat{\nabla}$ is
defined by Eq.~$($\ref{eq:speconn}$)$ and $\omega_p$ is defined by
Eq.~$($\ref{eq:syform}$)$.\label{pp:motivation}
\end{prop}
\begin{proof}
As pointed out in Section~\ref{se:cotan}, by~\cite{Bai1, Bai2},
$\hat{\nabla}$ is a left invariant flat and torsion free connection
on $T^*G$ which is endowed with the Lie group structure defined by
Eq.~(\ref{eq:coliestru}). So we only need to prove that $\omega_p$
is parallel with respect to $\hat{\nabla}$. In fact, for any
$x,y,z\in\frak{g},a^*,b^*,c^*\in\frak{g}^*$, we have that
\begin{eqnarray*}
\omega_p(\hat{\nabla}_{(x,a^*)}(y,b^*),(z,c^*))&=&\omega_p((\nabla_xy,
-b^*\circ\nabla_x),(z,c^*))=-\langle\nabla_xy,c^*\rangle-\langle
b^*,\nabla_xz\rangle,\\
\omega_p(\hat{\nabla}_{(x,a^*)}(z,c^*),(y,b^*))&=&\omega_p((\nabla_xz,-c^*\circ\nabla_x),
(y,b^*))=-\langle\nabla_xz,b^*\rangle-\langle c^*,\nabla_xy\rangle.
\end{eqnarray*}
So $\omega_p(\hat{\nabla}_{(x,a^*)}(y,b^*),(z,c^*))=
\omega_p(\hat{\nabla}_{(x,a^*)}(z,c^*),(y,b^*))$.
\end{proof}

\subsection{Post-left-symmetric algebras and the affine cotangent extension problem}
Proposition~\ref{pp:motivation} motivates us to consider the
following {\it affine cotangent extension problem} (cf.~\cite{Bo2}):

{\it Let $G$ be a Lie group whose Lie algebra is $\frak{g}$. Suppose
that $\nabla$ is a left invariant affine structure on $G$. Then
$\frak{g}$ is equipped with a compatible LSA structure defined by
Eq.~$($\ref{eq:delsa}$)$. Our aim is to find all the LSA structures
on $\frak{g}\times\frak{g}^*$  satisfying the following conditions:
\begin{enumerate}
\item
\begin{equation*}
\xymatrix{ 0\ar@{->}[r]&
\frak{g}^*\ar@{->}^{i}[r]&\frak{g}\times\frak{g}^*\ar@{->}^{p}[r]&\frak{g}\ar@{->}[r]&0
}
\end{equation*}
 is an exact sequence of LSAs, where
$\frak{g}^*$ is equipped with the trivial LSA structure. Here $i$
and $p$ are the canonical inclusion and projection respectively,
that is, $i(a^*)=(0,a^*)\in\frak{g}\times\frak{g}^*$ and
$p((x,a^*))=x\in\frak{g}$ for any $x\in\frak{g},a^*\in\frak{g}^*$.
\item
The bilinear form $\omega_p$ defined by Eq.~$($\ref{eq:syform}$)$
induces a left invariant symplectic form $($on the corresponding Lie
group whose Lie algebra is the sub-adjacent Lie algebra of the LSA
structure on $\frak{g}\times\frak{g}^*$$)$ which is parallel with
respect to the flat and torsion free connection induced by the LSA
structure on $\frak{g}\times\frak{g}^*$.
\end{enumerate}
}

We will show that the solution of the affine cotangent extension
problem is related to the following new algebraic structure.

\begin{defn}
{\rm A {\it post-left-symmetric algebra $($PLSA$)$}
$(A,\prec,\succ)$ is a vector space $A$ equipped with two bilinear
operations $\prec,\succ:A\otimes A\rightarrow A $ such that
$(A,\succ)$ is an LSA, $\prec$ is {\it commutative}, that is,
$x\prec y=y\prec x$, for any $x,y\in A$, and the following
compatibility condition holds:
\begin{equation}
x\succ(y\prec z)=(x\cdot y)\prec z+y\prec(x\cdot z),\quad \forall
x,y,z\in A,\label{eq:compatible}
\end{equation}
where $x\cdot y=x\prec y+x\succ y$. }
\end{defn}

\begin{prop}
Let $A$ be a vector space equipped with two bilinear operations
$\prec,\succ:A\otimes A\to A$. Suppose that $\prec$ is commutative
and $\prec$ and $\succ$ satisfy Eq.~$($\ref{eq:compatible}$)$, where
$x\cdot y=x\prec y+x\succ y$ for any $x,y\in A$. Then $(A,\succ)$ is
an LSA if and only if $(A,\cdot)$ is an LSA. Therefore in this case,
 $(A,\prec,\succ)$ is a PLSA if and only if
$(A,\cdot)$ is an LSA. Furthermore, if $(A,\prec,\succ)$ is a PLSA,
then the sub-adjacent
 Lie algebras of $(A,\cdot)$ and $(A,\succ)$
 coincide.
 \label{pp:pplsa}
\end{prop}

\begin{proof}
By Eq.~(\ref{eq:compatible}), for any $x,y,z\in A$, we have
\begin{equation}
[x,y]\prec z=x\succ(y\prec z)-y\succ(x\prec z)-y\prec(x\cdot
z)+x\prec(y\cdot z),\label{eq:compa}
\end{equation}
where $[x,y]=x\cdot y-y\cdot x$. Since the operation $\prec$ is
commutative, $(A,\succ)$ is an LSA if and only if the following
equation holds:
\begin{equation}
[x,y]\succ z=x\succ (y\succ z)-y\succ(x\succ z),\label{eq:compa1}
\end{equation}
where $[x,y]=x\cdot y-y\cdot x$. By Eq.~(\ref{eq:compa}), we show
that Eq.~(\ref{eq:compa1}) holds if and only if the following
equation holds:
$$[x,y]\cdot z=x\succ(y\cdot z)-y\succ(x\cdot z)-y\prec(x\cdot
z)+x\prec(y\cdot z)=x\cdot(y\cdot z)-y\cdot(x\cdot z).$$ So
$(A,\succ)$ is an LSA if and only if $(A,\cdot)$ is an LSA. Hence in
this case, $(A,\prec,\succ)$ is a PLSA if and only if $(A,\cdot)$ is
an LSA. The last conclusion follows from the fact that the operation
$\prec$ is commutative.
\end{proof}

The last conclusion of Proposition~\ref{pp:pplsa} motivates us to
give the following definitions.

\begin{defn}
{\rm Let $(A,\prec,\succ)$ be a PLSA. We denote the LSA structure
$(A,\cdot)$ by $l(A)$ which is called the {\it associated LSA} of
$(A,\prec,\succ)$. $(A,\prec,\succ)$ is called a {\it compatible
PLSA} on $l(A)$. Moreover, the sub-adjacent Lie algebra
$\frak{g}(A)$ of the two LSAs $(A,\cdot)$ and $(A,\succ)$ is called
{\it the sub-adjacent Lie algebra} of $(A,\prec,\succ)$ and
$(A,\prec,\succ)$ is called a {\it compatible PLSA} on
$\frak{g}(A)$. On the other hand, $(A,\prec,\succ)$ is said to be a
{\it PLSA on the LSA $(A,\succ)$}. }
\end{defn}

At the Lie group level we have

\begin{defn}
{\rm Let $G$ be a Lie group with Lie algebra $\frak{g}$. A {\it left
invariant post-affine structure} on $G$ is given by a pair
$\{\nabla,\tilde{\nabla}\}$ of left invariant flat and torsion free
connections which are compatible in the sense that
\begin{equation}
\nabla_x(\tilde{\nabla}_yz-\nabla_yz)=(\tilde{\nabla}_z-\nabla_z)\tilde{\nabla}_xy+
(\tilde{\nabla}_y-\nabla_y)\tilde{\nabla}_xz,\quad \forall
x,y,z\in\frak{g}.\label{eq:poaffco}
\end{equation}
}
\end{defn}

Let $G$ be a connected and simply connected Lie group whose Lie
algebra is $\frak{g}$. Suppose that $\{\nabla,\tilde{\nabla}\}$ is a
left invariant post-affine structure on $G$. Then it is easy to
check that the following bilinear operations define a compatible
PLSA structure on $\frak{g}$:
\begin{equation}
x\succ y=\nabla_xy,\;\; x\prec y=\tilde{\nabla}_xy-\nabla_xy,\;\;
\forall x,y\in\frak{g}.
\end{equation}
Note that since both $\nabla$ and $\tilde{\nabla}$ are torsion free,
the operation ``$\prec$" defined as above is automatically
commutative. Conversely, if $(\frak{g},\prec,\succ)$ is a compatible
PLSA structure on $\frak{g}$, then, after left translating, the
following connections (on $\frak{g}$) induce a left invariant
post-affine structure on $G$:
\begin{equation}
\nabla_xy=x\succ y,\quad \tilde{\nabla}_xy=x\cdot y=x\prec y+x\succ
y,\quad \forall x,y\in\frak{g}.
\end{equation}
So we have the following one-to-one correspondence:
\begin{equation*}
\xymatrix{ \text{\{left invariant post-affine structures on $G$\}}
\ar@{<->}[r]& \text{\{compatible PLSA structures on $\frak{g}$\}} }.
\end{equation*}

One can compare it with the correspondence formulated by
Eq.~(\ref{eq:onetoone}). In general, we can define a {\it
post-affine structure on an $n$-dimensional smooth manifold $M$} to
be a pair of flat and torsion free connections
$\{\nabla,\tilde{\nabla}\}$ such that Eq.~(\ref{eq:poaffco}) holds,
where $x,y,z$ are vector fields on $M$.

Returning to the affine cotangent extension problem,  it is obvious
that the LSA structure on $\frak{g}\times\frak{g}^*$ can be written
as follows:
\begin{equation}
(x,a^*)\circ(y,b^*)=(x\cdot y,l(x)b^*+r(y)a^*+\varphi(x,y)),\quad
\forall x,y\in\frak{g},a^*,b^*\in\frak{g}^*,\label{eq:affinepro}
\end{equation}
where $l,r:\frak{g}\to\frak{gl}(\frak{g}^*)$ and
$\varphi:\frak{g}\otimes\frak{g}\to\frak{g}^*$ are linear maps.

\begin{lemma}
Let $\frak{g}$ be a Lie algebra with a compatible LSA structure. Let
$l,r:\frak{g}\to\frak{gl}(\frak{g}^*)$ and
$\varphi:\frak{g}\otimes\frak{g}\to\frak{g}^*$ be three linear maps.
Then Eq.~$($\ref{eq:affinepro}$)$ defines an LSA structure on
$\frak{g}\times\frak{g}^*$ if and only if (for any
$x,y,z\in\frak{g}$)
\begin{equation} l(x)l(y)-l(x\cdot y)=l(y)l(x)-l(y\cdot
x),\label{eq:lsabimod1}
\end{equation}
\begin{equation}
l(x)r(y)-r(y)l(x)=r(x\cdot y)-r(y)r(x),\label{eq:lsabimod2}
\end{equation}
\begin{equation}
r(z)\varphi(x,y)+\varphi(x\cdot
y,z)-l(x)\varphi(y,z)-\varphi(x,y\cdot
z)=r(z)\varphi(y,x)+\varphi(y\cdot
x,z)-l(y)\varphi(x,z)-\varphi(y,x\cdot z).\label{eq:vaicohomo}
\end{equation}
\label{le:circlsa}
\end{lemma}

\begin{proof}
In fact, for any $x,y,z\in\frak{g},a^*,b^*,c^*\in\frak{g}^*$, we
have
\begin{eqnarray*}
& &((x,a^*)\circ(y,b^*))\circ(z,c^*)=(x\cdot
y,l(x)b^*+r(y)a^*+\varphi(x,y))\circ(z,c^*)\\ &=&((x\cdot y)\cdot
z,l(x\cdot
y)c^*+r(z)l(x)b^*+r(z)r(y)a^*+r(z)\varphi(x,y)+\varphi(x\cdot y,z)).
\end{eqnarray*}
On the other hand,
\begin{eqnarray*}
& &(x,a^*)\circ((y,b^*)\circ(z,c^*))\\&=&(x\cdot(y\cdot
z),l(x)l(y)c^*+l(x)r(z)b^*+l(x)\varphi(y,z)+r(y\cdot
z)a^*+\varphi(x,y\cdot z)),\\
& &((y,b^*)\circ(x,a^*))\circ(z,c^*)\\&=&((y\cdot x)\cdot z,l(y\cdot
x)c^*+r(z)l(y)a^*+r(z)r(x)b^*+r(z)\varphi(y,x)+\varphi(y\cdot x,z))\\
& &(y,b^*)\circ((x,a^*)\circ(z,c^*))\\&=&(y\cdot(x\cdot
z),l(y)l(x)c^*+l(y)r(z)a^*+l(y)\varphi(x,z)+r(x\cdot
z)b^*+\varphi(y,x\cdot z)).
\end{eqnarray*}
So it is easy to show that the operation ``$\circ$"  defines an LSA
if and only if Eq.~(\ref{eq:lsabimod1}), Eq.~(\ref{eq:lsabimod2})
and Eq.~(\ref{eq:vaicohomo}) hold.
\end{proof}

Keep the notations above. Define two new bilinear operations
$\prec,\succ:\frak{g}\otimes\frak{g}\to\frak{g}$ by
\begin{equation}
\langle x\prec y,a^*\rangle=-\langle r(x)a^*,y\rangle,\quad x\succ
y=x\cdot y-x\prec y,\quad x,y\in\frak{g},a^*\in\frak{g}^*.
\end{equation}
So we have $r(x)=L_{\prec}^*(x)$.

\begin{theorem}
With the same conditions and notations as above, the affine
cotangent extension problem has a solution if and only if the
following conditions hold:
\begin{enumerate}
\item
$l=L_{\cdot}^*$.
\item
$(\frak{g},\prec,\succ)$ is a PLSA.
\item
$\varphi$ satisfies Eq.~$($\ref{eq:vaicohomo}$)$ and the following
equation:
\begin{equation}
\langle\varphi(x,y),z\rangle=\langle\varphi(x,z),y\rangle,\quad
\forall x,y\in\frak{g}.
\end{equation}
\end{enumerate}
\label{thm:soofpr}
\end{theorem}
\begin{proof}
In fact, for any $x,y,z\in\frak{g},a^*,b^*,c^*\in\frak{g}^*$, we
have
\begin{eqnarray*}
\omega_p((x,a^*)\circ(y,b^*),(z,c^*))&=&-\langle x\cdot
y,c^*\rangle+\langle l(x)b^*+r(y)a^*+\varphi(x,y),z\rangle\\
\omega_p((x,a^*)\circ(z,c^*),(y,b^*))&=&-\langle x\cdot
z,b^*\rangle+\langle l(x)c^*+r(z)a^*+\varphi(x,z),y\rangle.
\end{eqnarray*}
So
$\omega_p((x,a^*)\circ(y,b^*),(z,c^*))=\omega_p((x,a^*)\circ(z,c^*),(y,b^*))$
if and only if the following conditions hold:
\begin{eqnarray*}
-\langle x\cdot y,c^*\rangle&=&\langle
l(x)c^*,y\rangle\Leftrightarrow l(x)=L_{\cdot}^*(x)\Leftrightarrow
\langle l(x)b^*,z\rangle=-\langle x\cdot
z,b^*\rangle,\\
\langle r(y)a^*,z\rangle&=&\langle r(z)a^*,y\rangle\Leftrightarrow
y\prec z=z\prec y,\\
 \langle\varphi(x,y),z\rangle&=&\langle
\varphi(x,z),y\rangle.
\end{eqnarray*}
On the other hand, since $(\frak{g},\cdot)$ is an LSA,
$l=L_{\cdot}^*$ automatically satisfies Eq.~(\ref{eq:lsabimod1}).
Furthermore, Eq.~(\ref{eq:lsabimod2}) holds if and only if
\begin{eqnarray*}
& &\langle l(x)r(y)a^*-r(y)l(x)a^*,z\rangle=\langle r(x\cdot
y)a^*-r(y)r(x)a^*,z\rangle,\quad \forall x,y,z\in\frak{g},a^*\in\frak{g}^*,\\
&\Leftrightarrow& \langle a^*,y\prec(x\cdot z)-x\cdot(y\prec
z)\rangle=\langle a^*,-(x\cdot y)\prec z-x\prec(y\prec
z)\rangle,\quad \forall x,y,z\in\frak{g},a^*\in\frak{g}^*,\\
&\Leftrightarrow& x\succ(y\prec z)=(x\cdot y)\prec z+y\prec(x\cdot
z),\quad \forall x,y,z\in\frak{g}.
\end{eqnarray*}
So  the conclusion follows from Proposition~\ref{pp:pplsa} and
Lemma~\ref{le:circlsa}.
\end{proof}

It would be interesting to ask whether $\varphi$ has a cohomological
explanation (cf.~\cite{Bo2}).

\begin{defn}
{\rm  Let $(A,\cdot)$ be an LSA and $V$ be a vector space. Let
$l,r:A\rightarrow \frak{gl}(V)$ be two linear maps. $V$ (or
$(V,l,r)$) is called a {\it bimodule} of $(A,\cdot)$ if
Eq.~(\ref{eq:lsabimod1}) and Eq.~(\ref{eq:lsabimod2}) hold.}
\end{defn}

In fact, according to~\cite{Sc}, $(V,l,r)$ is a bimodule of an LSA
$(A,\cdot)$ if and only if the direct sum $A\oplus V$ of the
underlying vector spaces of $A$ and $V$ is turned into an LSA (the
{\it semidirect sum}) by defining multiplication in $A\oplus V$ by
\begin{equation}
(x_1+v_1)\cdot_1(x_2+v_2)=x_1\cdot x_2+(l(x_1)v_2+r(x_2)v_1),\quad
x_1,x_2\in A, v_1, v_2\in V.\end{equation} We denote this
left-symmetric algebraic structure by $A\ltimes_{l,r}V$.

From the proof of Theorem~\ref{thm:soofpr}, we have the following
conclusion:

\begin{coro}
If $(A,\prec,\succ)$ is a PLSA, then $(A^*,L_{\cdot}^*,L_{\prec}^*)$
is a bimodule of $l(A)$.\label{co:bimolsadu}
\end{coro}

Now let $(A,\prec,\succ)$ be a PLSA and let $G$ be the connected and
simply connected Lie group corresponding to $\frak{g}(A)$. Then we
have
$\frak{g}(l(A)\ltimes_{L_{\cdot}^*,L_{\prec}^*}A^*)=\frak{g}(A)\ltimes_{L_{\succ}^*}A^*$.
Now we equip $T^*G$ with the Lie group structure defined by
Eq.~(\ref{eq:coliestru}), where the left invariant connection
$\nabla$ is induced by the LSA structure ``$\succ$". Then the Lie
algebra of $T^*G$ is $\frak{g}(A)\ltimes_{L_{\succ}^*}A^*$. If we
set $\varphi=0$, then we have the following consequence of
Theorem~\ref{thm:soofpr}:

\begin{coro}
With the same conditions and notations as above, then the natural
skew-symmetric and nondegenerate bilinear form $\omega_p$  on
$\frak{g}(A)\ltimes_{L_{\succ}^*}A^*$ defined by
Eq.~$($\ref{eq:syform}$)$ induces a left invariant symplectic
structure on $T^*G$ such that it is parallel with respect to the
left invariant affine structure $\tilde{\nabla}$ induced by the
(compatible) LSA structure
$l(A)\ltimes_{L_{\cdot}^*,L_{\prec}^*}A^*$ (on
$\frak{g}(A)\ltimes_{L_{\succ}^*}A^*$), that is,
$(T^*G,\tilde{\nabla},\omega_p)$ is a special symplectic Lie
group.\label{co:plsaspe}
\end{coro}

\subsection{Special symplectic Lie groups and post-affine structures}
The following conclusion shows that the PLSA is the natural
underlying algebraic structure of a special symplectic Lie algebra.
\begin{prop}
Let $(G,\nabla,\omega)$ be a special symplectic Lie group whose Lie
algebra is $\frak{g}$. The left invariant affine structure $\nabla$
induces a compatible LSA structure $\cdot$ on $\frak{g}$ by
Eq.~$($\ref{eq:delsa}$)$. Then the following operations define a
compatible PLSA structure on $\frak{g}$:
\begin{equation}
\omega(x\prec y,z)=-\omega(y,z\cdot x),\quad \omega(x\succ
y,z)=\omega(y,[z,x]),\quad \forall
x,y,z\in\frak{g}.\label{eq:complsa}
\end{equation}
Hence it induces a left invariant post-affine structure on
$G$.\label{pp:complsa}
\end{prop}

\begin{proof}
Define a linear map $\varphi:\frak{g}\to\frak{g}^*$ through
Eq.~(\ref{eq:delimap}). Then for any $x,y,z\in\frak{g}$,
\begin{eqnarray*}
\omega(x\prec y,z)&=&-\omega(y,z\cdot x)\Leftrightarrow\langle
\varphi(x\prec y),z\rangle=\langle
R_{\cdot}^*(x)\varphi(y),z\rangle,\\
\omega(x\succ y,z)&=&\omega(y,[z,x])\Leftrightarrow\langle
\varphi(x\succ y),z\rangle =\langle{\rm ad}^*(x)\varphi(y),z\rangle.
\end{eqnarray*}
So the operations $\prec$ and $\succ$ are in fact defined by
$$x\prec y=\varphi^{-1}(R_{\cdot}^*(x)\varphi(y)),\quad x\succ
y=\varphi^{-1}({\rm ad}^*(x)\varphi(y)).$$ Since
$$\omega(x\prec y,z)=-\omega(y,z\cdot x)=-\omega(x,z\cdot
y)=\omega(y\prec x,z),$$  we have $x\prec y=y\prec x$, that is,
$\prec$ is commutative. Moreover, for any $w\in\frak{g}$,
\begin{eqnarray*}
\omega(x\succ(y\prec z),w)&=&\omega(y\prec
z,[w,x])=-\omega(z,[w,x]\cdot y),\\
\omega((x\cdot y)\prec z,w)&=&-\omega(z,w\cdot(x\cdot y)),\\
\omega(y\prec(x\cdot z),w)&=&-\omega(x\cdot z,w\cdot
y)=\omega(z,x\cdot (w\cdot y)).
\end{eqnarray*}
Since $(\frak{g},\cdot)$ is an LSA, we have
$$x\succ(y\prec z)=(x\cdot y)\prec z+y\prec(x\cdot z).$$
Similarly, we show that $(\frak{g},\succ)$ is an LSA and $x\cdot
y=x\prec y+x\succ y$ for any $x,y\in \frak{g}$.
\end{proof}

\begin{exam}
{\rm The 2-dimensional (compatible) PLSAs $(\frak
g,\prec_i,\succ_i)$ corresponding to the special symplectic Lie
algebras $(\frak g,\nabla^i,\omega)$ ($i=1,2,3,4$) in Example 2.2
are given by
\begin{eqnarray*}
&& {\rm (I)}\;\; e_i\prec_1 e_j=0,\;\;e_i\succ_1 e_j=0,\;\;i,j=1,2;\\
&& {\rm (II)}\;\; e_1\prec_2 e_1=e_2,\;\;
e_1\prec_2e_2=e_2\prec_2e_1=e_2\prec_2e_2=0,\;\;
e_i\succ_2e_j=0\;\;i,j=1,2;\\
&& {\rm (III)}\;\; e_1\prec_3
e_1=e_2\prec_3e_2=0,e_1\prec_3e_2=e_2\prec_3e_1= -e_1,\\
&&\hspace{1cm} e_1\succ_3e_1=e_2\succ_3=0, e_1\succ_3
e_2=e_1,e_2\succ_3e_2=e_2;\\
&& {\rm (IV)}\;\; e_1\prec_4
e_1=0,\;e_1\prec_4e_2=e_2\prec_4e_1= -\frac{1}{2}e_1,e_2\prec_4e_2=e_1-\frac{1}{2}e_2,\\
&&\hspace{1cm} e_1\succ_4e_1=e_2\succ_4=0, e_1\succ_4
e_2=e_1,e_2\succ_4e_2=e_2,
\end{eqnarray*}
respectively. }\end{exam}

\begin{remark}
{\rm Since $\omega$ is also a left invariant symplectic form on $G$,
we can define a left invariant affine structure $\nabla$ on $G$
through Eq.~(\ref{eq:sycoles11}). It is obvious that its
corresponding LSA structure coincides with the LSA structure $\succ$
defined by Eq.~(\ref{eq:complsa}). Moreover, it interprets the
commuting diagram introduced by Eq.~(\ref{eq:diagram}). }
\label{re:explain}
\end{remark}

\subsection{Double extensions of special symplectic Lie algebras (groups)}
\label{ss:doubleex} In this subsection we investigate the
``nonabelian extension" of the construction in
Corollary~\ref{co:plsaspe}. More precisely, we consider the
following construction. Let $G$ be a connected and simply connected
Lie group whose Lie algebra is $\frak g$. Suppose that there exists
a left invariant affine structure $\nabla$ on $G$ which induces a
compatible LSA structure ``$\circ_{\frak{g}}$" on $\frak{g}$. Now
assume that there is an LSA structure ``$\circ_{\frak{g}^*}$" on the
dual space $\frak{g}^*$ which induces a left invariant affine
structure $\nabla^*$ on the connected and simply connected Lie group
$G^*$ corresponding to the sub-adjacent Lie algebra of the LSA
$(A^*,\circ_{\frak{g}^*})$. We consider how to construct an LSA
structure ``$\circ$" on a direct sum $\frak{g}\oplus \frak{g}^*$ of
the underlying vector spaces of $\frak{g}$ and $\frak{g}^*$ such
that $(\frak{g},\circ_{\frak{g}})$ and
$(\frak{g}^*,\circ_{\frak{g}^*})$ are subalgebras and the natural
skew-symmetric and nondegenerate bilinear form $\omega_p$ defined by
Eq.~(\ref{eq:syform}) is parallel with respect to the left invariant
affine structure induced by the LSA structure $\circ$ on the
connected and simply connected Lie group $\mathcal{D}(G)$
corresponding to the sub-adjacent Lie algebra of the LSA
$(\frak{g}\oplus \frak{g}^*,\circ)$. Such a construction is called
{\it a double extension of special symplectic Lie algebras}
associated to $(\frak{g},\circ_{\frak{g}})$ and
$(\frak{g}^*,\circ_{\frak{g}^*})$. It is denoted by
$(\frak{g}\bowtie \frak{g}^*,\circ,\omega_p)$. On the Lie group
level it is called {\it a double extension of special symplectic Lie
groups}. It is denoted by $(\mathcal{D}(G),\hat{\nabla},\omega_p)$,
where $\hat{\nabla}$ is the left invariant affine structure on the
``double Lie group" $\mathcal{D}(G)$ induced by the compatible LSA
structure $\circ$ on $\frak{g}\bowtie\frak{g}^*$. Furthermore, in
this case $(\mathcal{D}(G),G,G^*)$ is a {\it local double Lie group}
in the sense of~\cite{LM}.

\begin{remark}
{\rm
\begin{enumerate}
\item
In fact, the notion of double extension of special symplectic Lie
groups is very similar to that of Drinfeld's ``double Poisson-Lie
group"~(\cite{CP, Dr}).
\item
Let $\frak{g}$ be a Lie algebra with a hypersymplectic structure
$\{J,E,g\}$. According to~\cite{An}, the connection $\nabla^{CP}$
determined by the complex product structure $\{J,E\}$ coincides with
the Levi-Civita connection $\nabla^{g}$ of the metric $g$ and the
three symplectic forms $\omega_1,\omega_2$ and $\omega_3$ defined by
Eq.~(\ref{eq:3form}) are all $\nabla^{g}$-parallel. So if the metric
$g$ is flat, then $(\frak{g},\omega_i,\circ)$ ($i=1,2,3$) are all
special symplectic Lie algebras, where $\frak{g}$ is equipped with
the LSA structure $\circ$ determined by the flat and torsion free
connection $\nabla^{g}$. On the other hand, by
Lemma~\ref{le:3formex}, we show that $(\frak{g},\omega_2)$ is a
para-K\"{a}hler Lie algebra, that is, a symplectic Lie algebra with
two Lagrange subalgebras: if $(\frak{g},\frak{g}_{+},\frak{g}_{-})$
is the double Lie algebra associated to the complex product
structure $\{J,E\}$, then the two Lagrange subalgebras are given by
$\frak{g}_{+}$ and $\frak{g}_{-}$ respectively. Furthermore
according to ~\cite{Bai2}, $\frak{g}_{-}$ can be identified as the
dual space of $\frak{g}_{+}$. So if $\frak{g}$ is a flat
hypersymplectic Lie algebra, then $(\frak{g},\circ,\omega_2)$ is a
double extension of special symplectic Lie algebras associated to
$(\frak{g}_{+},\circ_{\frak{g}_{+}})$ and
$(\frak{g}_{-}=\frak{g}_{+}^*,\circ_{\frak{g}_{-}})$, where the LSA
structures $\circ_{\frak{g}_{+}}$ and $\circ_{\frak{g}_{-}}$ are
obtained by the restrictions of the LSA structure $\circ$ on
$\frak{g}_{+}$ and $\frak{g}_{-}$ respectively. So we come to the
following conclusion: the double extensions of special symplectic
Lie algebras are the natural underlying structures of flat
hypersymplectic Lie algebras.
\end{enumerate}
}\label{re:flathy}
\end{remark}

Let us recall the notion of {\it matched pair of LSAs} (\cite{Bai2})
which is parallel to the notions of {\it matched pair of Lie
algebras}, {\it matched pair of (Lie) groups} and {\it matched pair
of Hopf algebras} ~(\cite{LM, Ma, T}).

\begin{prop-def}  Let $(A_1,\cdot_1)$ and $(A_2,\cdot_2)$
be two LSAs. Suppose that there are linear maps
$l_{1},r_{1}:A_1\rightarrow \frak{gl}(A_2)$ and
$l_{2},r_{2}:A_2\rightarrow \frak{gl}(A_1)$ such that
$(A_2,l_{1},r_{1})$ is a bimodule of $A_1$ and $(A_1,l_{2},r_{2})$
is a bimodule of $A_2$ and they satisfy
\begin{equation}
r_{2}(a)([x,y]_1)=r_{2}(l_{1}(y)a)x-r_{2}(l_{1}(x)a)y+
x\cdot_1(r_{2}(a)y)-y\cdot_1(r_{2}(a)x),\label{eq:lsabi1}
\end{equation}
\begin{equation}
l_{2}(a)(x\cdot_1y)=-l_{2}(l_{1}(x)a-r_{1}(x)a)y+
(l_{2}(a)x-r_{2}(a)x)\cdot_1y+r_{2}(r_{1}(y)a)x+
x\cdot_1(l_{2}(a)y),\label{eq:lsabi2}
\end{equation}
\begin{equation}
r_{1}(x)([a,b]_2)=r_{1}(l_{2}(b)x)a-r_{1}(l_{2}(a)x)b+
a\cdot_2(r_{1}(x)b)-b\cdot_2(r_{1}(x)a),\label{eq:lsabi3}
\end{equation}
\begin{equation}
l_{1}(x)(a\cdot_2b)=-l_{1}(l_{2}(a)x-r_{2}(a)x)b+
(l_{1}(x)a-r_{1}(x)a)\cdot_2b+r_{1}(r_{2}(b)x)a+
a\cdot_2(l_{1}(x)b),\label{eq:lsabi4}
\end{equation}
where $x,y\in A_1, a,b\in A_2$ and $[\;,\;]_i$ is the sub-adjacent
Lie bracket of $(A, \cdot_i)$ ($i=1,2$). Then there is an LSA
structure ``$\cdot$" on the vector space $A_1\oplus A_2$ given by
\begin{equation}
(x+a)\cdot(y+b)=(x\cdot_1y+l_{2}(a)y+r_{2}(b)x)+(a\cdot_2b+
l_{1}(x)b+r_{1}(y)a),\;\; \forall x,y\in A_1, a,b\in A_2.
\end{equation}
We denote this LSA by $A_1\bowtie_{l_{1},r_{1}}^{l_{2},r_{2}}A_2$.
Moreover, $(A_1,A_2,l_{1},r_{1},l_{2},r_{2})$ satisfying the above
conditions is called a matched pair of LSAs. On the other hand,
every LSA which is a direct sum of the underlying vector spaces of
two subalgebras can be obtained from the above way.
\label{pd:matlsa}
\end{prop-def}

Our starting point is the following conclusion:

\begin{prop}
Let $(A,\prec,\succ)$ be a PLSA. Suppose that there exists a PLSA
structure on the dual space $A^*$ which we still denote by
``$\prec,\succ$". Then there exists a double extension of special
symplectic Lie algebras associated to $(\frak{g}(A),\cdot)$ and
$(\frak{g}(A^*),\cdot)$ if and only if $(l(A),l(A^*),L_{\cdot}^*$,
$L_{\prec}^*,L_{\cdot}^*,L_{\prec}^*)$ is a matched pair of
LSAs.\label{pp:startpo}
\end{prop}

\begin{proof}
In fact, if
$(l(A),l(A^*),L_{\cdot}^*,L_{\prec}^*,L_{\cdot}^*,L_{\prec}^*)$ is a
matched pair of LSAs, then the LSA structure on
$l(A)\bowtie_{L_{\cdot}^*,L_{\prec}^*}^{L_{\cdot}^*,L_{\prec}^*}l(A^*)$
can be written as follows:
\begin{equation}
(x,a^*)\cdot(y,b^*)=(x\cdot
y+L_{\cdot}^*(a^*)y+L_{\prec}^*(b^*)x,a^*\cdot
b^*+L_{\cdot}^*(x)b^*+L_{\prec}^*(y)a^*),\forall x,y\in A,a^*,b^*\in
A^*.\label{eq:lsatruonmat}
\end{equation}
For any $z\in A,c^*\in A^*$, we have
\begin{eqnarray*}
& &\omega_p((x,a^*)\cdot(y,b^*),(z,c^*))\\&=&\omega_p((x\cdot
y+L_{\cdot}^*(a^*)y+L_{\prec}^*(b^*)x,a^*\cdot
b^*+L_{\cdot}^*(x)b^*+L_{\prec}^*(y)a^*),(z,c^*))\\
&=&-\langle x\cdot y,c^*\rangle+\langle y,a^*\cdot c^*\rangle
+\langle x,b^*\prec c^*\rangle +\langle a^*\cdot
b^*,z\rangle-\langle b^*,x\cdot z\rangle-\langle a^*,y\prec
z\rangle,\\ & &\omega_p((x,a^*)\cdot(z,c^*),(y,b^*))\\
&=&\omega_p((x\cdot z+L_{\cdot}^*(a^*)z+L_{\prec}^*(c^*)x,a^*\cdot
c^*+L_{\cdot}^*(x)c^*+L_{\prec}^*(z)a^*),(y,b^*))\\
&=&-\langle x\cdot z,b^*\rangle+\langle z,a^*\cdot b^*\rangle +
\langle x,c^*\prec b^*\rangle+\langle a^*\cdot c^*,y\rangle-\langle
c^*,x\cdot y\rangle -\langle a^*,z\prec y\rangle.
\end{eqnarray*}
Hence $\omega_p((x,a^*)\cdot(y,b^*),(z,c^*))=
\omega_p((x,a^*)\cdot(z,c^*),(y,b^*))$. So $(\frak{g}(A)\bowtie
\frak{g}(A^*),\omega_p)$ is a double extension of special symplectic
Lie algebras associated to $(\frak{g}(A),\cdot)$ and
$(\frak{g}(A^*),\cdot)$.

Conversely, if there exists a double extension of special symplectic
Lie algebras associated to $(\frak{g}(A),\cdot)$ and
$(\frak{g}(A^*),\cdot)$, then by Proposition~\ref{pp:complsa}, there
exists a compatible PLSA structure on $(\frak{g}(A)\bowtie
\frak{g}(A^*),\omega_p)$ given by Eq.~(\ref{eq:complsa}). Moreover,
it is easy to show that $A$ and $A^*$ are post-left-symmetric
subalgebras. Set
$$x\cdot a^*=l_A(x)a^*+r_{A^*}(a^*)x,\quad a^*\cdot
x=l_{A^*}(a^*)x+r_{A}(x)a^*,\quad \forall x\in A,a^*\in A^*,$$ where
$l_A,r_A:A\to\frak{gl}(A^*)$, $l_{A^*},r_{A^*}:A^*\to\frak{gl}(A)$.
Then by Proposition-Definition~\ref{pd:matlsa},
$(l(A),l(A^*),l_A,r_A,l_{A^*},r_{A^*})$ is a matched pair of LSAs.
Moreover, for any $y\in A$, we have
$$\langle r_A(x)a^*,y\rangle=\omega_p(a^*\cdot
x,y)=-\omega_p(y,a^*\cdot x)=\omega_p(x\prec y,a^*)=-\langle
a^*,x\prec y\rangle\Rightarrow r_A=L_{\prec}^*,$$
\begin{eqnarray*}
\langle l_A(x)a^*,y\rangle&=&\omega_p(x\cdot
a^*,y)=\omega_p([x,a^*],y)+\omega_p(a^*\cdot x,y)=\omega_p(x\succ
y,a^*)+\langle r_A(x)a^*,y\rangle\\
&=&-\langle x\succ y,a^*\rangle-\langle a^*,x\prec y\rangle=\langle
y,L_{\cdot}^*(x)a^*\rangle\Rightarrow l_A=L_{\cdot}^*.
\end{eqnarray*}
Similarly or by symmetry of $A$ and $A^*$, we know that
$l_{A^*}=L_{\cdot}^*, r_{A^*}=L_{\prec}^*$. So
$(l(A),l(A^*),L_{\cdot}^*,\\ L_{\prec}^*,L_{\cdot}^*,L_{\prec}^*)$
is a matched pair of LSAs.
\end{proof}
Let $V$ be a vector space. Let $\sigma: V\otimes V \rightarrow
V\otimes V$ be the {\it exchanging operator} defined as
\begin{equation}\sigma(x \otimes y) = y\otimes x,\quad
\forall x, y\in V.
\end{equation}
Let $V_1,V_2$ be two vector spaces and $T:V_1\rightarrow V_2$ be a
linear map. Denote the dual (linear) map by $T^*:V_2^*\rightarrow
V_1^*$ defined by
\begin{equation}
\langle  v_1,T^*(v_2^*)\rangle   =\langle  T(v_1),v_2^*\rangle
,\;\;\forall v_1\in V_1, v_2^*\in V_2^*.
\end{equation}
Note that the notations of $L_\diamond^*$ and $R_\diamond^*$ given
by Eq.~(\ref{eq:starlrno1}) are different from the above notation if
$\frak{gl}(V)$ is regarded as a vector space. Furthermore note that
$(V\otimes V)^*=V^*\otimes V^*$ due to the assumption that $V$ is
finite-dimensional.

\begin{prop}
Let $(A,\prec,\succ,\alpha,\beta)$ be a PLSA $(A,\prec,\succ)$
endowed with two linear maps $\alpha,\beta:A\to A\otimes A$. Suppose
that $\alpha^*,\beta^*:A^*\otimes A^*\to A^*$ induce a PLSA
structure on $A^*$ which we still denote by $(\prec,\succ)$. Then
$(l(A),l(A^*),L_{\cdot}^*,L_{\prec}^*,L_{\cdot}^*,L_{\prec}^*)$ is a
matched pair of LSAs if and only if for any $x,y\in A$, the
following equations hold:
\begin{equation}
\alpha([x,y])=({\rm id}\otimes L_{\cdot}(x)+L_{\cdot}(x)\otimes{\rm
id})\alpha(y)-({\rm id}\otimes L_{\cdot}(y)+L_{\cdot}(y)\otimes{\rm
id})\alpha(x),\label{eq:1colsa1}
\end{equation}
\begin{equation}
(\alpha+\beta)(x\cdot y)=(L_{\succ}(x)\otimes{\rm id}+{\rm
id}\otimes L_{\cdot}(x))(\alpha+\beta)(y)+({\rm id}\otimes
R_{\cdot}(y))\beta(x)-(L_{\prec}(y)\otimes{\rm
id})\alpha(x),\label{eq:1colsa2}
\end{equation}
\begin{equation}\begin{matrix}
(\alpha+\beta-\sigma\alpha-\sigma\beta)(x\prec
y)=-(R_{\prec}(y)\otimes{\rm id})(\sigma\alpha+\sigma\beta)(x)+({\rm
id}\otimes R_{\prec}(y))(\alpha+\beta)(x)+\cr ({\rm id}\otimes
L_{\prec}(x))(\alpha+\beta)(y)-(L_{\prec}(x)\otimes{\rm
id})(\sigma\alpha+\sigma\beta)(y),
\cr\end{matrix}\mbox{}\hspace{0.5cm}\label{eq:1colsa3}
\end{equation}
\begin{equation}
(\alpha+\beta)(x\cdot y)=(L_{\succ}(x)\otimes{\rm id}+{\rm
id}\otimes L_{\cdot}(x))(\alpha+\beta)(y)+({\rm id}\otimes
R_{\cdot}(y))\beta(x)-(R_{\prec}(y)\otimes{\rm
id})\sigma\alpha(x).\label{eq:1colsa4}
\end{equation}
\label{pp:mato1co}
\end{prop}

\begin{proof}
By Corollary~\ref{co:bimolsadu} and
Proposition-Definition~\ref{pd:matlsa}, we need to prove that
Eq.~(\ref{eq:lsabi1})-Eq.~(\ref{eq:lsabi4}) are equivalent to
Eq.~(\ref{eq:1colsa1})-Eq.~(\ref{eq:1colsa4}) respectively
 in the case that
 $$A_1=A,\; A_2=A^*,\; {\rm and}\;
 l_1=L_{\cdot}^*,\; r_1=L_{\prec}^*,\; l_2=L_{\cdot}^*,\; r_2=L_{\prec}^*.$$
As an example we give an explicit proof of the fact that
Eq.~(\ref{eq:lsabi3}) holds if and only if Eq.~(\ref{eq:1colsa3})
holds. In fact, in this case, Eq.~(\ref{eq:lsabi3}) becomes
$$L_{\prec}^*(x)[a^*,b^*]=L_{\prec}^*(L_{\cdot}^*(b^*)x)a^*-L_{\prec}^*(L_{\cdot}^*(a^*)x)b^*+
a^*\cdot(L_{\prec}^*(x)b^*)-b^*\cdot(L_{\prec}^*(x)a^*),$$ for any
$x\in A^*,a^*,b^*\in A^*$. Let the both sides of the above equation
act on an arbitrary element $y\in A$. Then we get
\begin{eqnarray*}
-\langle[a^*,b^*],x\prec y\rangle&=&-\langle
b^*\cdot(R_{\prec}^*(y)a^*),x\rangle +\langle
a^*\cdot(R_{\prec}^*(y)b^*),x\rangle+\langle
a^*\cdot(L_{\prec}^*(x)b^*),y\rangle\\ & &-\langle
b^*\cdot(L_{\prec}^*(x)a^*),y\rangle,
\end{eqnarray*}
which is equivalent to the following equation
\begin{eqnarray*}
& &\langle a^*\otimes
b^*,(\alpha+\beta-\sigma\alpha-\sigma\beta)(x\prec
y)\rangle\\&=&\langle-(R_{\prec}(y)\otimes{\rm
id})(\sigma\alpha+\sigma\beta)(x)+({\rm id}\otimes
R_{\prec}(y))(\alpha+\beta)(x)+({\rm id}\otimes
L_{\prec}(x))(\alpha+\beta)(y)\\ & &-(L_{\prec}(x)\otimes{\rm
id})(\sigma\alpha+\sigma\beta)(y),a^*\otimes b^*\rangle.
\end{eqnarray*}
It exactly gives Eq.~(\ref{eq:1colsa3}).
\end{proof}

\begin{defn}{\rm
\begin{enumerate}
\item
Let $V$ be a vector space and $\alpha,\beta:V\to V\otimes V$ be two
linear maps. Then $(V,\alpha,\beta)$ is called a {\it
post-left-symmetric coalgebra $($PLSCA$)$} if $R_1,R_2$ and $R_3$
are all zero, where  $R_i:V\to V\otimes V\otimes V$ ($i=1,2,3$) are
three linear maps defined as follows (for any $x\in V$):
\begin{equation}
R_1(x)=\alpha(x)-\sigma\alpha(x),\label{eq:der1}
\end{equation}
\begin{equation}
R_2(x)=({\rm id}\otimes\alpha)\beta(x)-((\alpha+\beta)\otimes{\rm
id})\alpha(x)-(\sigma\otimes{\rm id})({\rm
id}\otimes(\alpha+\beta))\alpha(x),\label{eq:der2}
\end{equation}
\begin{equation}
R_3(x)=(\beta\otimes{\rm id})\beta(x)-(\sigma\otimes{\rm
id})(\beta\otimes{\rm id})\beta(x)-({\rm
id}\otimes\beta)\beta(x)+(\sigma\otimes{\rm id})({\rm
id}\otimes\beta)\beta(x).\label{eq:der3}
\end{equation}
It is obvious that $(V,\alpha,\beta)$ is a PLSCA if and only if
$(\alpha^*,\beta^*)$ induces a PLSA structure on $V^*$.
\item
Let  $(A,\prec,\succ,\alpha,\beta)$ be a PLSA with two linear maps
$\alpha,\beta:A\to A\otimes A$ such that $(A,\alpha,\beta)$ is a
PLSCA. If in addition, $\alpha$ and $\beta$ satisfy
Eq.~(\ref{eq:1colsa1})-Eq.~(\ref{eq:1colsa4}), then
$(A,\prec,\succ,\alpha,\beta)$ is called a {\it post-left-symmetric
bialgebra $($PLSBA$)$}.
\end{enumerate}
}\end{defn}

PLSBAs can be put into the framework of the {\it generalized
bialgebras} in the sense of Loday~(\cite{Lo}).

\begin{defn}
{\rm A post-left-symmetric bialgebra $(A,\prec,\succ,\alpha,\beta)$
is called {\it coboundary} if $\alpha$ and $\beta$ are given by the
following equations:
\begin{equation}
\alpha(x)=({\rm id}\otimes L_{\cdot}(x)+L_{\cdot}(x)\otimes{\rm
id})r,\label{eq:coboundary1}
\end{equation}
\begin{equation}
\beta(x)=(-{\rm id}\otimes{\rm ad}(x)-L_{\succ}(x)\otimes{\rm
id})r,\label{eq:coboundary2}
\end{equation}
where $x\in A, r\in A\otimes A$.}
\end{defn}

\begin{defn}
{\rm Let $(A,\prec,\succ,\alpha,\beta)$ be a coboundary
post-left-symmetric bialgebra. Then by Proposition~\ref{pp:startpo}
and Proposition~\ref{pp:mato1co}, there exists a double extension of
special symplectic Lie algebras
$(\frak{g}(A)\bowtie\frak{g}(A^*),\omega_p)$ associated to
$(\frak{g}(A),\cdot)$ and $(\frak{g}(A^*),\cdot)$, where
$\alpha^*,\beta^*:A^*\otimes A^*\to A^*$ induce a PLSA structure on
$A^*$ and we still denote the LSA structure of $l(A^*)$ by $\cdot$.
In this case, we say $(\frak{g}(A)\bowtie\frak{g}(A^*),\omega_p)$ is
a {\it double extension of special symplectic Lie algebras on $A$}.}
\label{de:douexon}
\end{defn}

\begin{prop}
Let $(A,\prec,\succ)$ be a PLSA. Let $\alpha,\beta:A\to A\otimes A$
be two linear maps defined by Eq.~$($\ref{eq:coboundary1}$)$ and
Eq.~$($\ref{eq:coboundary2}$)$ respectively. Then $\alpha$ and
$\beta$ satisfy
Eq.~$($\ref{eq:1colsa1}$)$-Eq.~$($\ref{eq:1colsa4}$)$ if and only if
the following two equations are satisfied:
\begin{equation}
(L_{\prec}(x\prec y)\otimes{\rm id}+{\rm id}\otimes L_{\prec}(x\prec
y)-L_{\prec}(y)\otimes L_{\prec}(x)-L_{\prec}(x)\otimes
L_{\prec}(y))(r-\sigma(r))=0.\label{eq:lecolsa1}
\end{equation}
\begin{equation}
(R_{\prec}(y)\otimes{\rm id})({\rm id}\otimes
L_{\cdot}(x)+L_{\cdot}(x)\otimes{\rm
id})(r-\sigma(r))=0.\label{eq:lecolsa2}
\end{equation}
for any $x,y\in A$.\label{pp:lecolsa}
\end{prop}
\begin{proof}
It is obvious that $\alpha$ and $\beta$ satisfy
Eq.~(\ref{eq:1colsa1}) and Eq.~(\ref{eq:1colsa2}). Moreover,
substituting Eq.~(\ref{eq:coboundary1}) and
Eq.~(\ref{eq:coboundary2}) into Eq.~(\ref{eq:1colsa3}), we get (for
any $x,y\in A$)
\begin{eqnarray*}
& &({\rm id}\otimes R_{\cdot}(x\prec y)+L_{\prec}(x\prec
y)\otimes{\rm id})r-(R_{\cdot}(x\prec y)\otimes{\rm id}+{\rm
id}\otimes
L_{\prec}(x\prec y))\sigma(r)\\
&=&-(R_{\prec}(y)\otimes{\rm id})(R_{\cdot}(x)\otimes{\rm id}+{\rm
id}\otimes L_{\prec}(x))\sigma(r)+({\rm id}\otimes
R_{\prec}(y))({\rm id}\otimes R_{\cdot}(x)+L_{\prec}(x)\otimes{\rm
id})r\\ & &+({\rm id}\otimes L_{\prec}(x))({\rm id}\otimes
R_{\cdot}(y)+L_{\prec}(y)\otimes{\rm id})r-(L_{\prec}(x)\otimes{\rm
id})(R_{\cdot}(y)\otimes{\rm id}+{\rm id}\otimes
L_{\prec}(y))\sigma(r),
\end{eqnarray*}
which is equivalent to
\begin{eqnarray*}
& &({\rm id}\otimes R_{\cdot}(x\prec y)+L_{\prec}(x\prec
y)\otimes{\rm id}-R_{\cdot}(x\prec y)\otimes{\rm id}-{\rm id}\otimes
L_{\prec}(x\prec y))r+\\ & &(R_{\cdot}(x\prec y)\otimes{\rm id}+{\rm
id}\otimes L_{\prec}(x\prec y))(r-\sigma(r))\\
&=&(R_{\prec}(y)\otimes{\rm id})(R_{\cdot}(x)\otimes{\rm id}+{\rm
id}\otimes L_{\prec}(x))(r-\sigma(r))-(R_{\prec}(y)\otimes{\rm
id})(R_{\cdot}(x)\otimes{\rm id}+\\ & &{\rm id}\otimes
L_{\prec}(x))r+({\rm id}\otimes R_{\prec}(y))({\rm id}\otimes
R_{\cdot}(x)+L_{\prec}(x)\otimes{\rm id})r+({\rm id}\otimes
L_{\prec}(x))({\rm id}\otimes R_{\cdot}(y)\\
& &+L_{\prec}(y)\otimes{\rm id})r-(L_{\prec}(x)\otimes{\rm
id})(R_{\cdot}(y)\otimes{\rm id}+{\rm id}\otimes
L_{\prec}(y))r+(L_{\prec}(x)\otimes{\rm id})(R_{\cdot}(y)\otimes{\rm
id}\\ & &+{\rm id}\otimes L_{\prec}(y))(r-\sigma(r)).
\end{eqnarray*}
It is easy to show that the above equation holds if and only if
Eq.~(\ref{eq:lecolsa1}) holds. So $\alpha$ and $\beta$ satisfy
Eq.~(\ref{eq:1colsa3}) if and only if Eq.~(\ref{eq:lecolsa1}) holds.
Similarly, we  show that $\alpha$ and $\beta$ satisfy
Eq.~(\ref{eq:1colsa4}) if and only if Eq.~(\ref{eq:lecolsa2}) holds.
\end{proof}

Let $V$ be a vector space and $r=\sum_ia_i\otimes b_i\in A\otimes
A$. Set
\begin{equation}
r_{12}=\sum_ia_i\otimes b_i\otimes 1,\quad r_{13}=\sum_{i}a_i\otimes
1\otimes b_i,\quad r_{23}=\sum_i1\otimes a_i\otimes b_i,
\end{equation}
where $1$ is a scale. If in addition, there exists a binary
operation $\diamond: V\otimes V\rightarrow V$ on $V$, then the
operation between two $r$s is in an obvious way. For example,
\begin{equation}
r_{12}\diamond r_{13}=\sum_{i,j}a_i\diamond a_j\otimes b_i\otimes
b_j,\; r_{13}\diamond r_{23}=\sum_{i,j}a_i\otimes a_j\otimes
b_i\diamond b_j,\;r_{12}\diamond r_{23}=\sum_{i,j}a_i\otimes
b_i\diamond  a_j\otimes b_j.
\end{equation}

\begin{prop}
Let $(A,\prec,\succ)$ be a PLSA and $r=\sum_{i}a_i\otimes b_i\in
A\otimes A$. Let $\alpha,\beta:A\to A\otimes A$ be two linear maps
defined by Eq.~$($\ref{eq:coboundary1}$)$ and
Eq.~$($\ref{eq:coboundary2}$)$ respectively. Define three linear
maps $R_1,R_2,R_3:A\to A\otimes A\otimes A$ by
Eq.~$($\ref{eq:der1}$)$, Eq.~$($\ref{eq:der2}$)$ and
Eq.~$($\ref{eq:der3}$)$ respectively. Then for any $x,y\in A$,
\begin{equation}
R_1(x)=(L_{\cdot}(x)\otimes{\rm id}+{\rm id}\otimes
L_{\cdot}(x))(r-\sigma(r)),\label{eq:tranlsaco1}
\end{equation}
\begin{equation}
R_2(x)=-Q_1(x)[[r,r]]_1+\sum_jP_1(x,a_j)(r-\sigma(r))\otimes
b_j,\label{eq:tranlsaco2}
\end{equation}
\begin{equation}
R_3(x)=Q_2(x)[[r,r]]_2+\sum_jP_2(x,a_j)(r-\sigma(r))\otimes
b_j+\sum_jS(a_j)(r-\sigma(r))\otimes [x,b_j],\label{eq:tranlsaco3}
\end{equation}
where
\begin{equation}
Q_1(x)=L_{\succ}(x)\otimes{\rm id}\otimes{\rm id}+{\rm id}\otimes
L_{\cdot}(x)\otimes{\rm id}+{\rm id}\otimes{\rm id}\otimes
L_{\cdot}(x),
\end{equation}
\begin{equation}
Q_2(x)=L_{\succ}(x)\otimes{\rm id}\otimes{\rm id}+{\rm id}\otimes
L_{\succ}(x)\otimes{\rm id}+{\rm id}\otimes{\rm id}\otimes {\rm
ad}(x),
\end{equation}
\begin{equation}
S(x)={\rm ad}(x)\otimes{\rm id}+{\rm id}\otimes L_{\succ}(x),
\end{equation}
\begin{equation}
P_1(x,y)=(R_{\prec}(y)\otimes{\rm id})(L_{\cdot}(x)\otimes{\rm
id}+{\rm id}\otimes L_{\cdot}(x)),
\end{equation}
\begin{equation}
P_2(x,y)={\rm ad}(x\succ y)\otimes{\rm id}+{\rm id}\otimes
L_{\succ}(x\succ y)-(R_{\succ}(y)\otimes{\rm id})({\rm
ad}(x)\otimes{\rm id}+{\rm id}\otimes L_{\succ}(x)),
\end{equation}
\begin{equation}
[[r,r]]_1=r_{13}\cdot r_{23}+r_{12}\cdot r_{23}+r_{12}\prec r_{13},
\end{equation}
\begin{equation}
[[r,r]]_2=r_{12}\succ r_{13}-r_{12}\succ r_{23}-[r_{13},r_{23}].
\end{equation}
\label{pp:tranlsaco}
\end{prop}

\begin{proof}
Eq.~(\ref{eq:tranlsaco1}) is obvious. We give an explicit proof of
Eq.~(\ref{eq:tranlsaco2}). The proof of Eq.~(\ref{eq:tranlsaco3}) is
similar (cf.~\cite{Bai2}). In fact, for any $x\in A$, after
rearranging the terms suitably, we can divide $R_2(x)$ into three
parts: $R_2(x)=(R1)+(R2)+(R3)$, where
\begin{eqnarray*}
(R1)&=&\sum_{i,j}-x\succ a_j\otimes b_j\cdot a_i\otimes b_i-x\succ
a_j\otimes a_i\otimes b_j\cdot b_i-(x\cdot a_j)\prec a_i\otimes
b_i\otimes b_j-\\ & &(x\cdot b_j)\prec a_i\otimes a_j\otimes b_i\\
&=&-(L_{\succ}(x)\otimes{\rm id}\otimes{\rm
id})[[r,r]]_1+\sum_j(R_{\prec}(a_j)\otimes{\rm
id})(L_{\cdot}(x)\otimes{\rm id})(r-\sigma(r))\otimes b_j,\\
(R2)&=&\sum_{i,j}-a_j\otimes[x,b_j]\cdot a_i\otimes b_i-a_i\otimes
b_i\cdot(x\cdot a_j)\otimes b_j-b_j\prec a_i\otimes x\cdot
a_j\otimes b_i-\\
& &a_i\otimes x\cdot a_j\otimes b_i\cdot b_j
\\ &=&-({\rm id}\otimes L_{\cdot}(x)\otimes{\rm
id})[[r,r]]_1+\sum_j(R_{\prec}(a_j)\otimes
L_{\cdot}(x))(r-\sigma(r))\otimes b_j,\\
(R3)&=&\sum_{i,j}-a_j\otimes a_i\otimes[x,b_j]\cdot b_i-a_j\prec
a_i\otimes b_i\otimes x\cdot b_j-a_i\otimes b_i\cdot a_j\otimes
x\cdot b_j-\\ & &a_i\otimes a_j\otimes b_i\cdot(x\cdot b_j)
\\&=&-({\rm id}\otimes{\rm id}\otimes
L_{\cdot}(x))[[r,r]]_1.
\end{eqnarray*}
So the conclusion follows.
\end{proof}

By Proposition~\ref{pp:lecolsa} and Proposition~\ref{pp:tranlsaco}
we have the following conclusion.

\begin{theorem}
Let $(A,\prec,\succ)$ be a PLSA and $r\in A\otimes A$. Then the
linear maps $\alpha,\beta$ defined by Eq.~$($\ref{eq:coboundary1}$)$
and Eq.~$($\ref{eq:coboundary2}$)$ respectively induce a
post-left-symmetric coalgebra structure on $A$ such that
$(A,\prec,\succ,\alpha,\beta)$ becomes a post-left-symmetric
bialgebra if and only if $r$ satisfies Eq.~$($\ref{eq:lecolsa1}$)$
and Eq.~$($\ref{eq:lecolsa2}$)$ and $R_1,R_2$ and $R_3$ are zero,
where $R_1,R_2$ and $R_3$ are given by
Eq.~$($\ref{eq:tranlsaco1}$)$, Eq.~$($\ref{eq:tranlsaco2}$)$ and
Eq.~$($\ref{eq:tranlsaco3}$)$ respectively.\label{thm:lsaconclusion}
\end{theorem}

A direct application of Theorem~\ref{thm:lsaconclusion} is given as
follows, which is an analogue of ``Drinfeld's double"
construction~(\cite{CP}) for a post-left-symmetric bialgebra. It is
the main motivation that we are interested in developing a
(coboundary) bialgebra theory for a PLSA since it allows us to
construct an infinite families of special symplectic Lie algebras
(groups) from a ``double extension of special symplectic Lie
algebras (groups)".

\begin{theorem}
Let $(A,\prec_1,\succ_1,\alpha,\beta)$ be a post-left-symmetric
bialgebra. Then there exists a canonical coboundary
post-left-symmetric bialgebra structure on $A\oplus A^*$. So if
$(\frak{g}\bowtie\frak{g}^*,\omega_p)$ is a double extension of
special symplectic Lie algebras associated to
$(\frak{g},\circ_{\frak{g}})$ and $(\frak{g}^*,\circ_{\frak{g}^*})$,
then we can construct a double extension of special symplectic Lie
algebras on $\frak{g}\bowtie\frak{g}^*$ $($see
Definition~\ref{de:douexon}$)$. \label{thm:double}
\end{theorem}

To prove this theorem, we shall use the following lemma.

\begin{lemma}
Let $(A,\prec,\succ)$ be a PLSA. Suppose that there exists a PLSA
structure on the dual space $A^*$ which is still denoted by
``$\prec,\succ$". Now we assume that there exists a double extension
of special symplectic Lie algebras
$(\frak{g}(A)\bowtie\frak{g}(A^*),\omega_p)$ associated to
$(\frak{g}(A),\cdot)$ and $(\frak{g}(A^*),\cdot)$. Then the
compatible PLSA structure on $\frak{g}(A)\bowtie\frak{g}(A^*)$
defined by Eq.~$($\ref{eq:complsa}$)$ can be given as follows:
\begin{equation}
x\prec a^*=R_{\cdot}^*(x)a^*+R_{\cdot}^*(a^*)x,\quad x\succ a^*={\rm
ad}^*(x)a^*-R_{\succ}^*(a^*)x,\label{eq:productwri1}
\end{equation}
\begin{equation}
a^*\prec x=R_{\cdot}^*(a^*)x+R_{\cdot}^*(x)a^*,\quad a^*\succ x={\rm
ad}^*(a^*)x-R_{\succ}^*(x)a^*,\label{eq:productwri2}
\end{equation}
for any $x,y\in A,a^*,b^*\in A^*$.\label{le:productwri}
\end{lemma}
\begin{proof}
In fact, for any $x,y\in A,a^*,b^*\in A^*$, we have
\begin{eqnarray*}
\langle x\prec a^*,b^*\rangle&=&-\omega_p(x\prec
a^*,b^*)=\omega_p(a^*,b^*\cdot x)=\langle
a^*,L_{\cdot}^*(b^*)x\rangle =\langle R_{\cdot}^*(a^*)x,b^*\rangle,\\
\langle x\prec a^*,y\rangle&=&\omega_p(x\prec
a^*,y)=-\omega_p(a^*,y\cdot x)=\langle R_{\cdot}^*(x)a^*,y\rangle.
\end{eqnarray*}
Hence $x\prec a^*=R_{\cdot}^*(x)a^*+R_{\cdot}^*(a^*)x$. Moreover,
\begin{eqnarray*}
\langle x\succ a^*,b^*\rangle&=&-\omega_p(x\succ
a^*,b^*)=-\omega_p(a^*,[b^*,x])=-\langle
a^*,L_{\succ}^*(b^*)x\rangle=\langle b^*,-R_{\succ}^*(a^*)x\rangle\\
\langle x\succ a^*,y\rangle&=&\omega_p(x\succ
a^*,y)=\omega_p(a^*,[y,x])=\langle{\rm ad}^*(x)a^*,y\rangle.
\end{eqnarray*}
Hence $x\succ a^*={\rm ad}^*(x)a^*-R_{\succ}^*(a^*)x$. By symmetry
of $A$ and $A^*$ we have
$$a^*\prec x=R_{\cdot}^*(a^*)x+R_{\cdot}^*(x)a^*,\quad a^*\succ
x={\rm ad}^*(a^*)x-R_{\succ}^*(x)a^*.$$
\end{proof}

\noindent {\it Proof of Theorem~\ref{thm:double}.} By
Proposition~\ref{pp:startpo}, Proposition~\ref{pp:mato1co} and
Lemma~\ref{le:productwri}, we show that there exists a PLSA
structure on $A\oplus A^*$ (which is denoted by ``$\prec,\succ$")
given by Eq.~(\ref{eq:productwri1}) and Eq.~(\ref{eq:productwri2}).
Let $r\in A\otimes A^*\subset(A\oplus A^*)\otimes(A\oplus A^*)$
correspond
 to the identity map ${\rm id}:A\rightarrow A$. Let
$\{e_1,...,e_s\}$ be a basis of $A$ and $\{e_1^*,...e_s^*\}$ be its
dual basis. Then $r=\sum\limits_ie_i\otimes e_i^*$. Next we prove
$r$ satisfies the conditions of Theorem~\ref{thm:lsaconclusion}. If
so, then
$$\alpha_{\mathcal{D}}(u)=({\rm id}\otimes L_{\cdot}(u)+L_{\cdot}(u)\otimes{\rm
id})r\;\; {\rm and}\;\; \beta_{\mathcal{D}}(u)=(-{\rm id}\otimes{\rm
ad}(u)-L_{\succ}(u)\otimes{\rm id})r
$$ induce a post-left-symmetric bialgebra structure on $(A\oplus A^*,\prec,\succ)$,
where $u\in A\oplus A^*$. In fact, we shall prove that
$[[r,r]]_i=0$, $i=1,2$, and
\begin{equation}
(L_{\cdot}(u)\otimes{\rm id}+{\rm id}\otimes
L_{\cdot}(u))(r-\sigma(r))=0,\label{eq:prodoubl1}
\end{equation}
\begin{equation}
(L_{\prec}(u\prec v)\otimes{\rm id}+{\rm id}\otimes L_{\prec}(u\prec
v)-L_{\prec}(u)\otimes L_{\prec}(v)-L_{\prec}(v)\otimes
L_{\prec}(u))(r-\sigma(r))=0,\label{eq:prodoubl2}
\end{equation}
\begin{equation}
({\rm id}\otimes L_{\succ}(u)+{\rm ad}(u)\otimes{\rm
id})(r-\sigma(r))=0,\label{eq:prodoubl3}
\end{equation}
for all $u,v\in A\oplus A^*$. First, we have
\begin{eqnarray*}
[[r,r]]_1&=&r_{13}\cdot r_{23}+r_{12}\cdot r_{23}+r_{12}\prec
r_{13}\\
&=&\sum_{i,j}e_i\otimes e_j\otimes e_i^*\cdot e_j^*+e_i\otimes
e_i^*\cdot e_j\otimes e_j^*+e_i\prec e_j\otimes e_i^*\otimes e_j^*\\
&=&\sum_{i,j}e_i\otimes e_j\otimes e_i^*\cdot
e_j^*+e_i\otimes\{\langle e_i^*,-e_j\prec e_k\rangle e_k^*+\langle
e_j,-e_i^*\cdot e_k^*\rangle e_k\}\otimes e_j^*+\\
& &e_i\prec e_j\otimes e_i^*\otimes e_j^*=0.
\end{eqnarray*}
Similarly, $[[r,r]]_2=0$. Next we prove that
Eq.~(\ref{eq:prodoubl2}) holds. By a similar proof, we can show that
Eq.~(\ref{eq:prodoubl1}) and Eq.~(\ref{eq:prodoubl3}) hold. In fact,
in this case, Eq.~(\ref{eq:prodoubl2}) is equivalent to
\begin{eqnarray*}
& &\sum_k(u\prec v)\prec e_k\otimes e_k^*-(u\prec v)\prec
e_k^*\otimes
e_k+e_k\otimes(u\prec v)\prec e_k^*-e_k^*\otimes(u\prec v)\prec e_k\\
& &-v\prec e_k\otimes u\prec e_k^*+v\prec e_k^*\otimes u\prec
e_k-u\prec e_k\otimes v\prec e_k^*+u\prec e_k^*\otimes v\prec e_k=0,
\end{eqnarray*}
for any $u,v\in A\oplus A^*$. We can prove the above equation in the
following cases: (I) $u,v\in A$; (II) $u\in A,v\in A^*$; (III) $u\in
A^*,v\in A$; (IV) $u,v\in A^*$. As an example, we give a proof of
the first case (the proof of other cases is similar). Let
$u=e_i,v=e_j$, then coefficient of $e_m\otimes e_n$ (for any $m,n$)
in the above equation is
\begin{eqnarray*}
& &\sum_k-\langle (e_i\prec e_j)\prec
e_n^*,e_m^*\rangle+\langle(e_i\prec e_j)\prec
e_m^*,e_n^*\rangle-\langle e_j\prec e_k,e_m^*\rangle\langle e_i\prec
e_k^*,e_n^*\rangle+\\ & &\langle e_j\prec e_k^*,e_m^*\rangle\langle
e_i\prec e_k,e_n^*\rangle-\langle e_i\prec e_k,e_m^*\rangle \langle
e_j\prec e_k^*,e_n^*\rangle+\langle e_i\prec
e_k^*,e_m^*\rangle\langle
e_j\prec e_k,e_n^*\rangle\\
&=&\sum_k\langle e_i\prec e_j,e_m^*\cdot e_n^*\rangle-\langle
e_i\prec e_j,e_n^*\cdot e_m^*\rangle +\langle e_j\prec
e_k,e_m^*\rangle\langle e_i,e_n^*\cdot e_k^*\rangle\\
& &-\langle e_j,e_m^*\cdot e_k^*\rangle \langle e_i\prec
e_k,e_n^*\rangle+\langle e_i\prec e_k,e_m^*\rangle\langle
e_j,e_n^*\cdot e_k^*\rangle-\langle e_i,e_m^*\cdot e_k^*\rangle
\langle e_j\prec
e_k,e_n^*\rangle\\
&=&\langle e_i\prec e_j,[e_m^*,e_n^*]\rangle-\langle
e_j\prec(L_{\cdot}^*(e_n^*)e_i),e_m^*\rangle-\langle
R_{\cdot}^*(L_{\prec}^*(e_i)e_n^*)e_j,e_m^*\rangle-\langle
e_i\prec(L_{\cdot}^*(e_n^*)e_j),e_m^*\rangle\\ & &-\langle
R_{\cdot}^*(L_{\prec}^*(e_j)e_n^*)e_i,e_m^*\rangle\\
&=&\langle e_i\prec e_j,[e_m^*,e_n^*]\rangle-\langle
e_j\prec(e_n^*\cdot e_i)+e_i\prec(e_n^*\cdot e_j),e_m^*\rangle\\
&=&\langle e_i\prec e_j,[e_m^*,e_n^*]\rangle-\langle
e_n^*\succ(e_i\prec e_j),e_m^*\rangle=0.
\end{eqnarray*}
Similarly, the coefficients of $e_m^*\otimes e_n, e_m\otimes e_n^*$
and $e_m^*\otimes e_n^*$ are all zero, too. The last conclusion
follows from Proposition~\ref{pp:startpo} and
Proposition~\ref{pp:mato1co}.\hfill $\Box$

\section*{Appendix: An example of a special para-K\"ahler
Lie algebra}
In this appendix, we introduce a notion of a very
special para-K\"ahler Lie algebra and investigate its structure.

\medskip

 \noindent{\bf Definition A1}
\quad  A {\it para-K\"ahler structure} on a Lie algebra $\frak g$ is
a symplectic form $\omega:\frak g\otimes \frak g\rightarrow \mathbb
R$ and a paracomplex structure $E:\frak g\rightarrow \frak g$ and
they are {\it compatible} in the sense that
\begin{equation}
\omega(E(x),E(y))=-\omega(x,y),\quad \forall x,y\in\frak g.
\end{equation}
A Lie algebra with a para-K\"ahler structure is called a {\it
para-K\"ahler Lie algebra}.
\medskip

It is obvious that a para-K\"ahler Lie algebra is equivalent to a
symplectic Lie algebra such that it is a direct sum of the
underlying vector spaces of two Lagrangian subalgebras~(\cite{Bai1,
Bai2, Kan}). A  {\it para-K\"ahler manifold} is a symplectic
manifold $(M,\omega)$ with a {\it para-complex structure} $E$ such
that
\begin{equation}
\omega(E(X),E(Y))=-\omega(X,Y),\quad \forall X,Y\in\Gamma(TM),
\end{equation}
where $\omega$ is the symplectic form~(\cite{Lib}). It is denoted by
$(M,\omega,E)$. Recall that a {\it para-complex structure} on a
$2n$-dimensional smooth manifold is an endomorphism field
$E\in\Gamma({\rm End}TM)$ such that
\begin{enumerate}
\item
$E^2={\rm id}$;
\item
$[E(X),E(Y)]+[X,Y]=E([E(X),Y]+[X,E(Y)])$, for any
$X,Y\in\Gamma(TM)$, that is, $E$ is integrable;
\item
let $TM^{+}$ and $TM^{-}$ be the eigendistributions associated to
the eigenvalues $+1$ and $-1$ of $E$, then we have ${\rm
dim}TM_p^{+}={\rm dim}TM_p^{-}$ for any $p\in M$.
\end{enumerate}
A Lie group whose Lie algebra is a para-K\"ahler Lie algebra is a
particular example of a {\it homogeneous para-K\"ahler
manifold}~(\cite{Kan}).
\medskip

 \noindent{\bf Definition A2}\quad
 Let $(\frak{g},\omega,E)$ be a para-K\"ahler Lie algebra.
Suppose that there exists a flat and torsion free connection
$\nabla$ on $\frak{g}$ such that $\nabla E$ is a symmetric
$(1,2)$-tensor field, that is, $(\nabla_xE)y=(\nabla_yE)x$ for any
$x,y\in\frak{g}$, and $\omega$ is parallel with respect to $\nabla$,
that is, $\nabla\omega=0$. Then the quadruple
$(\frak{g},\omega,E,\nabla)$ is called a {\it special para-K\"ahler
Lie algebra}.
\medskip

 Note that a special para-K\"ahler Lie algebra is a
special symplectic Lie algebra.

A  {\it special para-K\"ahler manifold} is defined as a
para-K\"ahler manifold $(M,\omega,E)$ with a flat and torsion free
connection $\nabla$ such that
\begin{enumerate}
\item
$\nabla$ is symplectic, that is, $\nabla\omega=0$ and
\item
$\nabla E$ is a symmetric $(1,2)$-tensor field, i.e.,
$(\nabla_XE)Y=(\nabla_YE)X$ for all $X,Y\in\Gamma(TM)$.
\end{enumerate}
Special para-K\"ahler geometry was introduced in~\cite{CMMS}. It
arises as one of the special geometries of Euclidean super-symmetry.
It is obvious that the Lie group whose Lie algebra is a special
para-K\"ahler Lie algebra is a (homogeneous) special para-K\"ahler
manifold.

Next we investigate when a double extension of special symplectic
Lie algebras becomes a special para-K\"ahler Lie algebra. In fact,
according to Proposition~\ref{pp:startpo}, a double extension of
special symplectic Lie algebras is equivalent to the fact that
$(l(A),l(A^*),L_{\cdot}^*,L_{\prec}^*,L_{\cdot}^*,L_{\prec}^*)$ is a
matched pair of LSAs, where $(A,\prec,\succ)$ is a PLSA. Moreover,
it is easy to show that
$(\frak{g}(l(A)\bowtie_{L_{\cdot}^*,L_{\prec}^*}^{L_{\cdot}^*,L_{\prec}^*}l(A^*)),\omega_p,E)$
is a para-K\"ahler Lie algebra, where $\omega_p$ is defined by
Eq.~(\ref{eq:syform}) and
$E:\frak{g}(l(A)\bowtie_{L_{\cdot}^*,L_{\prec}^*}^{L_{\cdot}^*,L_{\prec}^*}l(A^*))\to
\frak{g}(l(A)\bowtie_{L_{\cdot}^*,L_{\prec}^*}^{L_{\cdot}^*,L_{\prec}^*}l(A^*))$
is the paracomplex structure on
$\frak{g}(l(A)\bowtie_{L_{\cdot}^*,L_{\prec}^*}^{L_{\cdot}^*,L_{\prec}^*}l(A^*))$
defined by
\begin{equation}
E(x+a^*)=x-a^*,\quad \forall x\in A,a^*\in
A^*.\label{eq:deprodcutst}
\end{equation}
Now let $\nabla$ be the connection corresponding to the compatible
LSA structure on
$\frak{g}(l(A)\bowtie_{L_{\cdot}^*,L_{\prec}^*}^{L_{\cdot}^*,L_{\prec}^*}l(A^*))$
defined by Eq.~(\ref{eq:lsatruonmat}). Then we have

\medskip

 \noindent{\bf Proposition A3}\quad\!\!
{\it  With the same conditions and notations as above,\!
$(\frak{g}(l(A)\bowtie_{L_{\cdot}^*,L_{\prec}^*}^{L_{\cdot}^*,L_{\prec}^*}l(A^*)),\omega_p,E,\!\nabla)$
is a special para-K\"ahler Lie algebra if and only if the operation
$\prec$ is trivial.}
\begin{proof}
By the proof of Proposition~\ref{pp:startpo}, we know that
$\omega_p$ is parallel with respect to $\nabla$. Moreover, for any
$x,y\in A,a^*,b^*\in A^*$, we have that
\begin{eqnarray*}
& &\nabla_{(x,a^*)}E((y,b^*))-E(\nabla_{(x,a^*)}(y,b^*))\\
&=&(x\cdot y+L_{\cdot}^*(a^*)y-L_{\prec}^*(b^*)x,-a^*\cdot
b^*-L_{\cdot}^*(x)b^*+L_{\prec}^*(y)a^*)-(x\cdot
y+L_{\cdot}^*(a^*)y+L_{\prec}^*(b^*)x,\\ & &-a^*\cdot
b^*-L_{\cdot}^*(x)b^*-L_{\prec}^*(y)a^*)\\
&=&(-2L_{\prec}^*(b^*)x,2L_{\prec}^*(y)a^*).
\end{eqnarray*}
Thus,
$\nabla_{(x,a^*)}E((y,b^*))-E(\nabla_{(x,a^*)}(y,b^*))=\nabla_{(y,b^*)}E((x,a^*))-E(\nabla_{(y,b^*)}(x,a^*))$
if and only if
$$(-2L_{\prec}^*(b^*)x,2L_{\prec}^*(y)a^*)=(-2L_{\prec}^*(a^*)y,2L_{\prec}^*(x)b^*),$$
for any $x,y\in A,a^*,b^*\in A^*$, if and only if the product
$\prec$ is trivial. So the conclusion follows immediately.
\end{proof}

Therefore we have the following obvious conclusion.

 \noindent{\bf Corollary A4}\quad
{\it  Let $(A,\cdot)$ be an LSA. Suppose that there exists an LSA
structure on $A^*$ $($which we still denote by $\cdot$$)$ such that
$(A,A^*,L_{\cdot}^*,0,L_{\cdot}^*,0)$ is a matched pair of LSAs.
Then
$(\frak{g}(A\bowtie_{L_{\cdot}^*,0}^{L_{\cdot}^*,0}A^*),\omega_p,E,\nabla)$
is a special para-K\"ahler Lie algebra, where $\omega_p$ is defined
by  Eq.~$($\ref{eq:syform}$)$ and $E$ is defined by
Eq.~$($\ref{eq:deprodcutst}$)$ and $\nabla$ is the connection on
$\frak{g}(A\bowtie_{L_{\cdot}^*,0}^{L_{\cdot}^*,0}A^*)$ induced by
the compatible LSA structure
$A\bowtie_{L_{\cdot}^*,0}^{L_{\cdot}^*,0}A^*$.\label{co:forspe} }
\medskip

 \noindent{\bf Definition A5}\quad  We call a special para-K\"ahler Lie algebra formulated in
Corollary A4 a {\it very special para-K\"ahler Lie algebra}.
\medskip

Our next task is to investigate the structure of a very special
para-K\"ahler Lie algebra. In fact, the following conclusion gives a
structure theory of a very special para-K\"ahler Lie algebra in
terms of a kind of bialgebras (cf.~\cite{Bai2}).
\medskip

 \noindent{\bf Theorem A6}\quad
{\it Let $(A,\cdot,\alpha)$ be an LSA endowed with a linear map
$\alpha:A\to A \otimes A$. Suppose that $\alpha^*:A^*\otimes A^*\to
A^*$ induces an LSA structure on $A^*$ which we still denote by
$\cdot$. Then $(A,A^*,L_{\cdot}^*,0,L_{\cdot}^*,0)$ is a matched
pair of LSAs if and only if the following equation holds:
\begin{equation}
\alpha(x\cdot y)=(L_{\cdot}(x)\otimes{\rm id}+{\rm id}\otimes
L_{\cdot}(x))\alpha(y)+({\rm id}\otimes
R_{\cdot}(y))\alpha(x),\label{eq:spelsba}
\end{equation}
for any $x,y\in A$.\label{thm:mato2co} }

\begin{proof}
By Proposition-Definition~\ref{pd:matlsa}, we need to prove that
Eq.~(\ref{eq:lsabi1})-Eq.~(\ref{eq:lsabi4}) are equivalent to
Eq.~(\ref{eq:spelsba})
 in the case that
 $$A_1=A,\;A_2=A^*,\; {\rm and}\;
 l_1=L_{\cdot}^*,\; r_1=0,\; l_2=L_{\cdot}^*,\; r_2=0.$$
In fact, it is easy to see that in this case, Eq.~(\ref{eq:lsabi1})
and Eq.~(\ref{eq:lsabi3}) hold automatically. Next, we prove that
Eq.~(\ref{eq:lsabi2})$\Leftrightarrow$Eq.~(\ref{eq:lsabi4})$\Leftrightarrow$
Eq.~(\ref{eq:spelsba}). As an example we give an explicit proof of
the fact that Eq.~(\ref{eq:lsabi2}) holds if and only if
Eq.~(\ref{eq:spelsba}) holds. The proof of the equivalence between
Eq.~(\ref{eq:lsabi4}) and Eq.~(\ref{eq:spelsba}) is similar. In
fact, in this case, Eq.~(\ref{eq:lsabi2}) becomes
$$L_{\cdot}^*(a^*)(x\cdot y)=-L_{\cdot}^*(L_{\cdot}^*(x)a^*)y+(L_{\cdot}^*(a^*)x)\cdot y+
x\cdot(L_{\cdot}^*(a^*)y),$$ for any $x,y\in A^*,a^*\in A^*$. Let
the both sides of the above equation act on an arbitrary element
$b^*\in A^*$. Then we get $$\langle x\cdot y,-a^*\cdot
b^*\rangle=\langle y,(L_{\cdot}^*(x)a^*)\cdot b^*\rangle+\langle
x,a^*\cdot(R_{\cdot}^*(y)b^*)\rangle+\langle
y,a^*\cdot(L_{\cdot}^*(x)b^*)\rangle,$$ which is equivalent to the
following equation $$\langle -\alpha(x\cdot y),a^*\otimes
b^*\rangle=\langle-(L_{\cdot}(x)\otimes{\rm id})\alpha(y)-({\rm
id}\otimes R_{\cdot}(y))\alpha(x)-({\rm id}\otimes
L_{\cdot}(x))\alpha(y),a^*\otimes b^*\rangle.$$ It exactly gives
Eq.~(\ref{eq:spelsba}).
\end{proof}

 \noindent{\bf Definition A7}\quad
\begin{enumerate}
\item
Let $V$ be a vector space and $\alpha:V\to V\otimes V$ be a linear
map. Then $(V,\alpha)$ is called a {\it left-symmetric coalgebra
$($LSCA$)$} if $T_{\alpha}$ is zero, where $T_{\alpha}:V\to V\otimes
V\otimes V$ is defined as follows (for any $x\in V$):
\begin{equation}
T_{\alpha}(x)=(\alpha\otimes{\rm id})\alpha(x)-(\sigma\otimes{\rm
id})(\alpha\otimes{\rm id})\alpha(x)-({\rm
id}\otimes\alpha)\alpha(x)+(\sigma\otimes{\rm id})({\rm
id}\otimes\alpha)\alpha(x).\label{eq:lsca}
\end{equation}
It is obvious that $(V,\alpha)$ is an LSCA if and only if
$\alpha^*:V^*\otimes V^*\to V^*$ induces an LSA structure on $V^*$.
\item
Let $(A,\cdot,\alpha)$ be an LSA with a linear map $\alpha:A\to
A\otimes A$ such that $(A,\alpha)$ is an LSCA. If in addition,
$\alpha$ satisfies Eq.~(\ref{eq:spelsba}), then $(A,\cdot,\alpha)$
is called a {\it special left-symmetric bialgebra $($SLSBA$)$}.
\end{enumerate}
\medskip

There is a notion of {\it left-symmetric bialgebra} which is
equivalent to the notion of para-K\"ahler Lie algebra~(\cite{Bai2}).
Like the post-left-symmetric bialgebras, both left-symmetric
bialgebras and special left-symmetric bialgebras can be put into the
framework of the {\it generalized bialgebras} in the sense of
Loday~(\cite{Lo}), too.
\medskip

 \noindent{\bf Definition A8}\quad
 A special left-symmetric bialgebra $(A,\cdot,\alpha)$ is called
{\it coboundary} if $\alpha$ is given by the following form:
\begin{equation}
\alpha(x)=({\rm id}\otimes R_{\cdot}(x))r,\label{eq:coslsab}
\end{equation}
where $x\in A,r\in A\otimes A$.
\medskip

 \noindent{\bf Proposition A9}\quad
{\it Let $(A,\cdot)$ be an LSA. Let $\alpha:A\to A\otimes A$ be a
linear map defined by Eq.~$($\ref{eq:coslsab}$)$. Then $\alpha$
satisfies Eq.~$($\ref{eq:spelsba}$)$ if and only if the following
equation is satisfied:
\begin{equation}
({\rm id}\otimes R_{\cdot}(y))(L_{\cdot}(x)\otimes{\rm id}+{\rm
id}\otimes L_{\cdot}(x))r=0,\;\;\forall x,y\in
A.\label{eq:specialco}
\end{equation}
\label{pp:specon11} }
\begin{proof}
Straightforward (cf. Proposition~\ref{pp:lecolsa}).
\end{proof}

 \noindent{\bf Proposition A10}\quad {\it Let $(A,\cdot)$ be an LSA and $r=\sum\limits_ia_i\otimes b_i\in
A\otimes A$. Let $\alpha:A\to A\otimes A$ be a linear map defined by
Eq.~$($\ref{eq:coslsab}$)$. Define a linear map $T_{\alpha}:A\to
A\otimes A$ by Eq.~$($\ref{eq:lsca}$)$. Then for any $x\in A$, we
have
\begin{equation}
T_{\alpha}(x)=({\rm id}\otimes{\rm id}\otimes
R_{\cdot}(x))(r_{12}\cdot r_{23}-r_{21}\cdot
r_{13}+[r_{13},r_{23}]),\label{eq:specit}
\end{equation}
where $r_{21}=\sum_ib_i\otimes a_i\otimes 1$ and $r_{21}\cdot
r_{13}=\sum_{i,j}b_i\cdot a_j\otimes a_i\otimes
b_j$.\label{pp:spelsco11}}

\begin{proof}
In fact, for any $x\in A$, we have

{\small \begin{eqnarray*} T_{\alpha}(x)&=&\sum_{i,j}a_i\otimes
b_i\cdot a_j\otimes b_j\cdot x-b_i\cdot a_j\otimes a_i\otimes
b_j\cdot x-a_j\otimes a_i\otimes
b_i\cdot (b_j\cdot x)+a_i\otimes a_j\otimes b_i\cdot(b_j\cdot x)\\
&=&({\rm id}\otimes{\rm id}\otimes R_{\cdot}(x))(r_{12}\cdot
r_{23}-r_{21}\cdot r_{13}+[r_{13},r_{23}]).
\end{eqnarray*}}
\end{proof}

By Proposition A9 and Proposition A10 we have the following result.
\medskip

 \noindent{\bf Theorem A11}\quad
{\it  Let $(A,\cdot)$ be an LSA and $r\in A\otimes A$. Then the
linear map $\alpha$ defined Eq.~$($\ref{eq:coslsab}$)$ induces an
LSCA structure on $A$ such that $(A,\cdot,\alpha)$ becomes a SLSBA
if and only if $r$ satisfies Eq.~$($\ref{eq:specialco}$)$ and
$T_{\alpha}$ given by Eq.~$($\ref{eq:specit}$)$ vanishes
identically.\label{thm:specconclsi} }
\medskip

 A direct application of
Theorem A11 is given as follows, which is parallel to
Theorem~\ref{thm:double}. It allows us to construct an infinite
family of very special para-K\"ahler Lie algebras from a very
special para-K\"ahler Lie algebra in a natural way. We would like to
point out that it can be regarded as a structure property of a very
special para-K\"ahler Lie algebra. Moreover, since a very special
para-K\"ahler Lie algebra is also a special symplectic Lie algebra,
it also enables us to construct an infinite family of special
symplectic Lie algebras from a very special para-K\"ahler Lie
algebra.
\medskip

 \noindent{\bf Theorem A12}\quad
{\it Let $(A,\cdot,\alpha)$ be a SLSBA. Then there exists a natural
coboundary SLSBA structure on $A\oplus A^*$.}

\begin{proof}
By Theorem A6, we know that $(A,A^*,L_{\cdot}^*,0,L_{\cdot}^*,0)$ is
a matched pair of LSAs. So there exists an LSA structure on $A\oplus
A^*$ given as follows:
$$(x,a^*)\cdot(y,b^*)=(x\cdot y+L_{\cdot}^*(a^*)y,a^*\cdot
b^*+L_{\cdot}^*(x)b^*),\;\;\forall x,y\in A,a^*,b^*\in A^*.$$
 Let $r\in A\otimes A^*\subset(A\oplus
A^*)\otimes(A\oplus A^*)$ correspond to the identity map ${\rm
id}:A\rightarrow A$. Let $\{e_1,...,e_s\}$ be a basis of $A$ and
$\{e_1^*,...e_s^*\}$ be its dual basis. Then
$r=\sum\limits_ie_i\otimes e_i^*$. Next we prove $r$ satisfies the
conditions of Theorem A11. If so, then
$\alpha_{\mathcal{SD}}(u)=({\rm id}\otimes R_{\cdot}(u))r$ induces a
coboundary SLABA structure on $(A\oplus A^*,\cdot)$, where $u\in
A\oplus A^*$. In fact, for any $k\in\{1,...,n\}$, we have
\begin{eqnarray*}
(L_{\cdot}(e_k)\otimes{\rm id}+{\rm id}\otimes
L_{\cdot}(e_k))\sum_ie_i\otimes e_i^*&=&\sum_ie_k\cdot e_i\otimes
e_i^*+e_i\otimes e_k\cdot e_i^*\\
&=&\sum_ie_k\cdot e_i\otimes e_i^*+e_i\otimes\langle e_i^*,-e_k\cdot
e_j\rangle e_j^*=0.
\end{eqnarray*}
Similarly, we have $(L_{\cdot}(e_k^*)\otimes{\rm id}+{\rm id}\otimes
L_{\cdot}(e_k^*))\sum\limits_ie_i\otimes e_i^*=0$. Therefore,
$(L_{\cdot}(u)\otimes{\rm id}+{\rm id}\otimes L_{\cdot}(u))r=0$ for
any $u\in A\oplus A^*$. Furthermore,
\begin{eqnarray*}
& &r_{12}\cdot r_{23}-r_{21}\cdot
r_{13}+[r_{13},r_{23}]\\&=&\sum_{i,j}e_i\otimes e_i^*\cdot
e_j\otimes e_j^*-e_i^*\cdot e_j\otimes e_i\otimes e_j+e_i\otimes
e_j\otimes
[e_i^*,e_j^*]\\
&=&\sum_{i,j}e_i\otimes\langle e_j,-e_i^*\cdot e_k^*\rangle
e_k\otimes e_j^*-\langle e_j,-e_i^*\cdot e_k^*\rangle e_k\otimes
e_i\otimes e_j+e_i\otimes e_j\otimes [e_i^*,e_j^*]=0.
\end{eqnarray*}
So the conclusion follows.
\end{proof}

\section*{Acknowledgements}

The first author thanks Professor N.J. Hitchin for kindly sending
his paper to him. This work was supported in part by the National
Natural Science Foundation of China (10621101, 10921061), NKBRPC
(2006CB805905) and SRFDP (200800550015).

\end{document}